\documentclass[a4paper,nofootinbib,preprintnumbers,floatfix,superscriptaddress,pra,twocolumn,showpacs]{revtex4-1}
\usepackage{amsmath, amsthm, amssymb}
\usepackage{graphicx}
\usepackage{dcolumn}
\usepackage{bm}
\usepackage{color}
\usepackage{amsmath}
\usepackage{amsfonts}
\usepackage{amssymb}
\usepackage{bbm}
\usepackage{hyperref}

\newcommand{\R}{\mathbb{R}}

\newcommand{\ket}[1]{| #1 \rangle}
\newcommand{\bra}[1]{\langle #1 |}

\newtheorem{prop}{Proposition}
\newtheorem{corollary}[prop]{Corollary}

\newtheorem{definition}[prop]{Definition}
\newtheorem{definitionlemma}[prop]{Definition and Lemma}

\newtheorem{lemma}[prop]{Lemma}

\newcommand{\beq}{\begin{equation}}
\newcommand{\eeq}{\end{equation}}
\newcommand{\bea}[1]{\begin{equation}\begin{array}{#1}}
\newcommand{\eea}{\end{array}\end{equation}}
\newcommand{\beqn}{\begin{eqnarray}}
\newcommand{\eeqn}{\end{eqnarray}}

\renewcommand{\rho}{\varrho}

\newcommand{\processnext}[1]{%
  \ifx\listfinish#1\empty\else\listact{#1}\expandafter\processnext\fi}

\newcommand{\eqnref}[1]{(\ref{#1})}

\newcommand{\ba}{\begin{eqnarray}}
\newcommand{\be}{\begin{equation}}
\newcommand{\ee}{\end{equation}}
\newcommand{\ea}{\end{eqnarray}}
\newcommand{\ban}{\begin{eqnarray*}}
\newcommand{\ean}{\end{eqnarray*}}

{\begin{framed}\begin{small}}
{\end{small}\end{framed}}

\DeclareGraphicsExtensions{.pdf,.png,.jpg}

\makeindex

\begin{document}

\title{Causal structures from entropic information: Geometry and novel scenarios}
\date{\today}

\author{Rafael Chaves}
\email{rafael.chaves@physik.uni-freiburg.de}
\affiliation{Institute for Physics, University of Freiburg, Rheinstrasse 10, D-79104 Freiburg, Germany}
\author{Lukas Luft}
\affiliation{Institute for Physics, University of Freiburg, Rheinstrasse 10, D-79104 Freiburg, Germany}
\author{David Gross}
\affiliation{Institute for Physics, University of Freiburg, Rheinstrasse 10, D-79104 Freiburg, Germany}

\begin{abstract}
The fields of quantum non-locality in physics, and causal discovery in machine learning, both face the problem of deciding whether observed data is compatible with a presumed causal relationship between the variables (for example a local hidden variable model). Traditionally, Bell inequalities have been used to describe the restrictions imposed by causal structures on marginal distributions. However, some structures give rise to non-convex constraints on the accessible data, and it has recently been noted that linear inequalities on the observable \emph{entropies} capture these situations more naturally.  In this paper, we show the versatility of the entropic approach by greatly expanding the set of scenarios for which entropic constraints are known. For the first time, we treat Bell scenarios involving multiple parties and multiple observables per party. Going beyond the usual Bell setup, we exhibit inequalities for scenarios with extra conditional independence assumptions, as well as a limited amount of shared randomness between the parties. Many of our results are based on a geometric observation: Bell
polytopes for two-outcome measurements can be naturally imbedded into the convex cone of attainable marginal entropies. Thus, any entropic inequality can be translated into one valid for probabilities. In some situations the converse also holds, which provides us with a rich source of candidate entropic inequalities.

\end{abstract}

\maketitle

\section{Introduction}

Starting point of this paper is the question: What can be inferred about the causal relationship of a collection of random variables from a restricted set of observations? To phrase this problem more precisely, we need to introduce the notions of a \emph{marginal scenario} and a \emph{causal structure} -- the two pieces of data which specify the instances we will be considering.

A marginal scenario describes which sets of random variables are
jointly observable. Joint observations might be constrained for a
variety of reasons. In quantum non-locality, these reasons are
physical: random variables corresponding to non-commuting observables
cannot always be jointly measured. In general, there might also be practical reasons, for instance: We have no access to the variable describing the genetic disposition of a patient to become both a smoker and to develop lung cancer (not the least because we do not know whether such a genetic influence exists).

For the purpose of this paper, a causal structure is a list of linear
constraints on the (conditional) mutual information between sets of
random variables. For example, in the familiar Bell scenario,
one commonly demands that
measurement choices of, respectively, Alice $X$ and Bob $Y$ are
independent of the hidden variable $\lambda$:
$I(X,Y:\lambda)=0$.
Relaxing this constraint to demand the correlations be small
$I(X,Y:\lambda)\leq \epsilon$ would still be linear in the mutual information
and thus an element of a causal structure according to our definition.
By allowing for arbitrary linear constraints, we go slightly beyond
the way the notion of ``causal structure'' is commonly formalized in
the field of causal inference \cite{Pearlbook,Spirtesbook}. There, the
combinatorial structure of a \emph{direct acyclic graph} (DAG) is used
to encode certain sets of conditional mutual informations that are
assumed to vanish. Our approach subsumes and extends this.

With every given causal structure, we can associate the set of
marginal distributions that are compatible with it. If we observe a
data point that lies outside that region, we can exclude the presumed
causal structure as a valid model for the observed data. This logical
structure (characterize the global properties compatible with local
observations) is an instance of a \emph{marginal problem}, which occur
frequently both in classical \cite{cuadras2002distributions} and in
quantum probability
\cite{Klyachko2006,Gross2012a,Gross2012b,Christandl2012}.

In quantum non-locality \cite{NLreview2013}, the focus has traditionally been on settings
for which the marginal distributions happen to be convex polytopes. In
that case, checking whether an observed marginal distribution is
compatible with the causal model reduces to the task of verifying that
none of the inequalities associated with the facets of the polytopes
is violated---these are the Bell inequalities~\cite{Bell1964,Pitowsky1989}.
However, the non-convex nature of mutual information means that the
marginals that appear for more general causal structures are, at best,
non-trivial algebraic varieties. A few such examples have been treated
in the quantum literature, including bilocality scenarios
\cite{Branciard2010} or scenarios that allow for correlations between Alice's and Bob's measurement choices with the hidden variable
\cite{Hall2010,Barrett2010} (c.f.\ also Section~\ref{sec:Imm}).

A priori, it is unclear whether these more complicated marginal
regions allow for an explicit description that is tractable from an
analytic and computational point of view. It is this problem that
entropic methods greatly simplify.  Indeed, as indicated above,
(conditional) independence constraints are \emph{linear} in terms of
Shannon entropies. As a result, the image of the marginal regions of
general causal structures turn out to possess natural descriptions in
terms of linear inequalities.

The set of all joint entropies (without any causal constraints and
prior to marginalization) has been analyzed extensively in information
theory \cite{Yeung1997,Yeung2008}. While it is known to be a convex
cone, its precise form is still not explicitly understood. For
practical purposes, it is often replaced by an outter approximation:
the convex \emph{Shannon cone}, which is defined by a finite number
of explicit \emph{Shannon type} inequalities. All marginals of the
Shannon cone can, in principle, be found computationally using linear
programming \cite{Yeung2008}. What is more, causal structures merely
amount to further linear constraints and can therefore be included in
a natural way.  An additional nice feature of entropic inequalities is
that they are valid for variables consisting of any number of
outcomes. This stands in stark contrast to the usual approach for
which increasing the number of outcomes of the marginal scenario
increases the dimension and complexity of the correlation polytope, in
practice meaning that new inequalities need to be derived and tailored
to the specific number of outcomes under consideration.

On the negative side entropic inequalities provide, in principle, only a necessary condition for the solution of the marginal problem~\cite{Chaves2013}. In spite of that, entropic inequalities are known to be fine enough to distinguish, for example, different causal structures~\cite{Steudel2010,Fritz2012} or witness non-locality and contextuality~\cite{Braunstein1988,Chaves2012,Kurzynki2012,FritzChaves2013,Pan2013,Devi2013,Katiyar2013,Kurzynski2013}.

In spite of its potential applications, the entropic approach to the marginal problem has been little explored. In particular, no entropic inequalities are known for Bell scenarios involving more than $2$ parties or many
measurement settings. Another problem, well suited to be tackled with entropies, is the one where the amount of shared randomness between the parties involved in a Bell test is bounded to be below a certain value. Commonly,
shared randomness is assumed to be a free and boundless resource, but quantitative considerations about how much of it is actually necessary to reproduce some quantum correlations can give useful insights that would be
extremely hard to tackle with the usual approaches.

These are the kind of problems we look at in this paper. In Sec. \ref{sec:shannon_cone} we start defining the entropic cone described by all Shannon-type inequalities. In Sec. \ref{sec:cones} we state known results about
convex cones that will be used in Sec. \ref{sec:backward_theorem} to prove a theorem showing that, for marginal scenarios without statistical independence, any Shannon-type inequality is also valid for the probabilities if a proper translation is made. We also show that the converse is in general
not true by providing a counter example showing that not every inequality for probabilities is also valid for Shannon entropies. Inspired by these results, in Sec. \ref{sec:Imm} we derive the entropic version of the
Collins-Gisin inequalities \cite{Collins2004} also considering the effects of bounded shared randomness between the parties. In Sec. \ref{sec:Multi} we derive a multipartite generalization of the entropic inequality
originally derived by Braunstein and Caves for the bipartite case \cite{Braunstein1988}. In Sec. \ref{sec:computational_results} we computationally apply the Fourier-Motzkin (FM) algorithm to derive entropic inequalities
for
a couple of different scenarios, including marginal models that also include statistical independencies and the effects of bounded shared randomness. We discuss our findings in Sec. \ref{sec:discussion} while technical
results and proofs can be found in the Appendices.

\section{Characterizing marginal scenarios with Shannon-type inequalities}
\label{sec:shannon_cone}

\subsection{Marginal Scenarios}

Given a set of variables $X_{1}, \dots, X_{n}$, a marginal scenario is a collection of certain subsets of them, those subsets of variables that can be jointly measured. In the case of Bell scenarios, the marginal scenario is
achieved by imposing space-like separation between some of the observables. Clearly, a subset of jointly measurable variables is still a jointly measurable set. Formally (see \cite{Chaves2012,FritzChaves2013} for further
details),

\begin{definition}
A \emph{marginal scenario} $\mathcal{M}$ is a collection
$\mathcal{M}=\{S_1,\ldots,S_{|\mathcal{M}|}\}$ of subsets $S_i\subseteq\{X_1,\ldots,X_n\}$ such that
if $S\in\mathcal{M}$ and $S'\subseteq S$, then also $S'\in\mathcal{M}$.
\end{definition}

In practice some joint statistics is measured for every
$S\in\mathcal{M}$. For example, if the variables $X_i$ and $X_j$ are
jointly measurable ($\{X_i,X_j\}=:S\in\mathcal{M}$), one can access
$P(X_i = x_i, X_j = x_j)$,  the probability of obtaining the outcomes
$X_i=x_i$ and $X_j=x_j$. These marginal probabilities determine in
particular
the
marginal Shannon entropy:
\begin{eqnarray*}
	&&
	H(X_i, X_j ) \\
	&=& - \sum_{x_i,x_{j}}
	P(X_i=x_i,X_j = x_{j})
	\log_{2}
	P(X_i=x_i,X_j = x_{j}).
\end{eqnarray*}
As first noticed in~\cite{Braunstein1988}, the existence of a joint
distribution for all variables $X_1,\cdots,X_n$ implies
that marginal Shannon entropies satisfy certain inequalities, which
may be violated by measurement statistics originating from quantum
experiments.
Below, we will recall how to compute these inequalities in general,
potentially in the presence of extra causal constraints.

\subsection{Entropy cones}

For the purpose of this section, assume that a number $n$ and some
joint distribution for $n$ random variables $X_1, \dots, X_n$. We
denote the set of indices of the random variables by $[n]=\{1, \dots,
n\}$ and its powerset (i.e.\ set of subsets) by $2^{[n]}$. For every
subset $S\in 2^{[n]}$ of indices, let $X_S$ be the tuple of
observables $(X_i)_{i\in S}$ and $H(S):=H(X_S)$ be the associated
marginal entropy. With this convention, the entropy becomes a function
\begin{eqnarray*}
	H: 2^{[n]} \to \mathbbm{R}, \qquad S \mapsto H(S)
\end{eqnarray*}
on the power set. The linear space of all set functions is of course
isomorphic to $\mathbbm{R}^{2^n}$ together with a basis $\{ e_S \,|\,|
S\subset 2^{[n]}\}$ labeled by subsets. We denote that vector space by
$R_n$ and will henceforth not distinguish between real-valued set
functions and the space $R_n$. For every vector $h\in\R_n$ and
$S\in2^{[n]}$, we denote by $h_S$ the component of $S$ with respect to
the basis vector $e_S$.

The region
\begin{eqnarray*}
	\left\{ h \in R_n \,|\, h_S = H(S) \text{ for some entropy function
	} H \right\}
\end{eqnarray*}
of vectors in $R_n$ that correspond to entropies has been researched
extensively in information theory \cite{Yeung2008}. It is known
to be a
convex cone (c.f.\ Section~\ref{sec:cones}), but an explicit
description has not yet been found. However, several properties of
entropy functions are well-understood. These are, respectively,
\emph{monotonicity}, \emph{sub-modularity}, and a normalization
condition:
\begin{align}
\begin{split}
	\label{shannonineqs_basic}
	H_{}(T\setminus\{i\}) &\leq H_{}(T) \\
	H_{}(S) + H_{}(S\cup\{i,j\}) &\leq H_{}(S\cup \{i\}) + H_{}(S\cup \{j\}) \\
	H_{}(\emptyset) &= 0
\end{split}
\end{align}
for all $S, T\in 2^{[n]}$.
The set of inequalities (\ref{shannonineqs_basic}) are known as the
\emph{elementary inequalities}
in information theory or the
\emph{polymatroidal axioms}.
An inequality that follows from the elementary ones is called a
\emph{Shannon-type inequality}.
The region defined by the Shannon-type inequalities is the
\emph{Shannon cone} $\Gamma_n$, a polyhedral
closed convex cone $\Gamma_n$. Clearly,  it is an outter approximation
to the true entropy cone. Since the latter is not yet fully
characterized, we will work for the remainder of this paper solely in
terms of the Shannon cone\footnote{
	This relaxation implies that while all inequalities we will derive
	below are valid for any true entropy vector, they may fail to be
	tight.
}. For future reference, we re-state this definition more formally:

\begin{definition}
	The \emph{Shannon cone} $\Gamma_n$
	is the set of vectors $h\in R_n$ that are
	\begin{enumerate}
        \item
		\emph{non-negative}
		\begin{equation*}
			h_{A} \geq 0 \text{ (with equality if $A=\emptyset$)},
		\end{equation*}

		\item
		\emph{increasing}
		\begin{equation*}
			h_B \leq h_A \text{ for $B \subseteq A$ },
		\end{equation*}

		\item
		\emph{sub-modular}
		\begin{equation*}
			h_{A\cup B} + h_{A \cap B} \leq h_{A} + h_{B}
		\end{equation*}
	\end{enumerate}
	for any $A, B \in 2^{[n]}$.
\end{definition}

We now return to descriptions involving a marginal scenario
$\mathcal{M}$.
Given a point $h:2^{[n]}\to\R$, one computes the restriction $h_{|\mathcal{M}}:\mathcal{M}\to\R$ dismissing the values $h(S)$ for all $S\in[n]\setminus\mathcal{M}$. The entropic cone bounding the correlations in
$\mathcal{M}$ is a projection of $\Gamma_n$ along a map $\R^{2^{[n]}}\to\R^\mathcal{M}$ throwing away some of the coordinates, that ones not corresponding to observable quantities. This set is also a convex cone, that we
denote by $\Gamma^{\mathcal{M}}$. Given an inequality description of $\Gamma^{\mathcal{M}}$, deciding if the marginal model can be extended is very simple, since one only needs to check whether it satisfies all the
inequalities defining it. In other terms, if a marginal model violates an inequality derived only by the combination of polymatroidal axioms, this implies that this marginal model cannot arise from a joint probability
distribution.

To determine the projection $\Gamma^{\mathcal{M}}$, a natural
possibility would be to calculate the extremal rays of $\Gamma_n$ and
dismiss the irrelevant coordinates of it. However, determining all the
extremal rays of the cone $\Gamma_n$ is a very hard problem, with
explicit solutions known only for few cases
\cite{Shapley72,Kashiwabara1996,Studeny2000}. To determine
$\Gamma^{\mathcal{M}}$ in practice we start with the inequality
description~(\ref{shannonineqs_basic}) of $\Gamma_n$ and then apply a
Fourier-Motzkin (FM) elimination \cite{Williams1986}, a standard
method for calculating the inequality description for the projection
of a polyhedral cone.

\subsection{Inequalities for marginal entropies}

To illustrate the general method, we begin considering the simplest non-trivial Bell scenario, corresponding to the CHSH scenario \cite{CHSH} and consisting of two parties, say A and B, who can measure one out of two
observables each, $\{A_0,A_1\}$ and $\{B_0,B_1\}$ respectively. This corresponds to a marginal scenario consisting of the following observable variables:
$\mathcal{M}=\{\{A_0,B_0\},\{A_0,B_1\},\{A_1,B_0\},\{A_1,B_1\}\}$ $=\{A_0,A_1,B_0,B_1,A_0B_0,A_0B_1,A_1B_0,A_1B_1\}$.
As shown in \cite{Chaves2012,FritzChaves2013} the only non-trivial Shannon-type entropic inequality (up to symmetries) corresponds to the inequality derived by Braunstein and Caves \cite{Braunstein1988}, the entropic CHSH,
given by
\begin{align}
\label{ECHSH}
CHSH^{E}= &-H_{A_0B_0}-H_{A_0B_1}-H_{A_1B_0}\\ \nonumber
&+H_{A_1B_1} + H_{A_0} + H_{B_0} \geq 0
\end{align}
where here and in following we employ the notation $H(A_iB_j)=H_{A_iB_j}$ (similarly to any number of variables) to avoid lengthy expressions.

In Ref. \cite{Braunstein1988} this inequality was derived using the chain rule of entropies. However, as just discussed, any Shannon-type inequality can be derived from the elemental set of
inequalities \eqref{shannonineqs_basic}. To illustrate the general procedure, we consider how to obtain the entropic inequality \eqref{ECHSH}, performing a FM elimination of the non-observable variables appearing in the set
of elementary inequalities. To derive the CHSH inequality \eqref{ECHSH} it is sufficient to combine the two \emph{sub-modularity inequalities}
\begin{align}
 &H_{A_0B_0}+H_{A_0B_1} \geq H_{A_0B_0B_1}+H_{A_0} \\
 &H_{A_1B_0}+H_{B_0B_1} \geq H_{A_1B_0B_1}+H_{B_0}.
\end{align}
Using that $H_{A_0B_0B_1} \geq H_{B_0B_1}$ and $H_{A_1B_0B_1} \geq H_{A_1B_1}$ we get exactly \eqref{ECHSH}. Note however, that these two last \emph{monotonicity inequalities} are not in the elemental set
\eqref{shannonineqs_basic}. To obtain for instance $H_{A_0B_0B_1} \geq H_{B_0B_1}$ from the basic ones we combine
\begin{align}
 &H_{A_0A_1B_0B_1} \geq H_{A_1B_0B_1}\\
 &H_{A_1B_0B_1}+H_{A_0B_0B_1} \geq H_{A_0A_1B_0B_1}+H_{B_0B_1}
\end{align}
It is clear that, in general, any monotonicity inequality, follows immediately from the basic ones.

One should note the similarity of $CHSH^{E}$ with the usual CHSH inequality in terms of probabilities \cite{CHSH,Collins2004}, that can be expressed as
\begin{align}
\label{CHSH}
CHSH= &q_{A_0B_0}+q_{A_0B_1}+q_{A_1B_0}\\ \nonumber
&-q_{A_1B_1} - q_{A_0} - q_{B_0} \leq 0
\end{align}
with $q_{A_iB_j}$ being the probability of getting the outcome $0$ if the measurement settings $i,j$ are used, and similarly for the marginals $q_{A_i}$ and $q_{B_j}$. We see that both inequalities are equivalent, if one
just makes the simple replacement $H_{A_iB_j} \rightarrow -q_{A_iB_j}$. Based on this simple observation we prove in Sec. \ref{sec:backward_theorem} a formal explanation to the similarities between the probability and entropic
inequalities.

\subsection{The role of causal structures}

Bell's theorem is usually associated with the incompatibility of
quantum correlations with a natural causal structure for space-like
separated events. However, in the derivation of the entropic
inequality \eqref{ECHSH} no explicit mention of a causal structure has
been made. Inequality \eqref{ECHSH} is valid for any set of $4$
variables. The only assumption made up to this point is the validity of
classical probability theory, or in other terms, the existence of a
well-defined joint probability distribution
$p(A_0=a_0,A_1=a_1,B_0=b_0,B_1=b_1)$. Bell's theorem can be seen as a
recipe  for interpreting the variables appearing in \eqref{ECHSH} or
\eqref{CHSH} as physically observable quantities.

We recall the usual argument:
Bell's theorem assumes a description of marginal models
where there exists a hidden variable $\lambda$
which
subsumes all the information the variables $A_0$, $A_1$, $B_0$,
and $B_1$, may depend on. This is the \emph{realism} assumption in
Bell's construction, assuring that all the variables have well-defined
values prior to any measurement. At each run of the experiment, Alice
and Bob independently chose which variable they will locally access,
tossing, respectively, uncorrelated coins $X$ and $Y$: if $X=0$ Alice
measures the observable associated with $A_0$, if $X=1$ she measures
$A_1$ (similarly to Bob). Because in general $A_0$ and $A_1$
(similarly $B_0$ and $B_1$) are associated with non-commuting
observables, quantum mechanics prohibits both to be jointly
measurable. The compatibility between $A_i$ and $B_j$ is guaranteed
by invoking the assumption of \emph{locality}, stating that space-like
events are not causally connected. Note however, that for example
$A_0$ is in principle not an observable quantity, rather what Alice
observes is $A_0$ conditioned on the fact that $X=0$. If $X$ is
correlated with $\lambda$, potentially the value of $A_0$ would be
different had Alice chosen to measure $A_1$. Here enters the final
assumption in Bell's theorem, that of \emph{measurement independence},
stating that $X$ and $Y$ are independent from the hidden variable
$\lambda$. Together, the three assumptions in Bell's theorem implies
the causal structure shown in Fig. \ref{fig:CHSHCS}.

\begin{figure} [!t]
\centering
\includegraphics[width=0.35\textwidth]{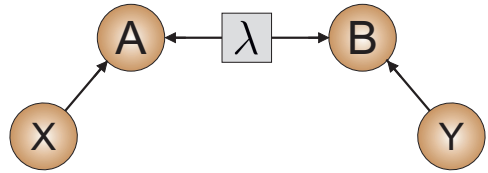}
\caption{
	\label{fig:CHSHCS}
Direct acyclic graph (DAG) representing the causal structure associated with a bipartite Bell experiment \cite{Bell1964}. The associated marginal scenario consists of all sets of variables that do not contain two different observables of the same party. For instance, if $X$ and $Y$ both are dichotomic with $x=0,1$ and $y=0,1$, we have the CHSH scenario \cite{Clauser1969} that is characterized by the marginal scenario $\mathcal{M}=\{\{A_0,B_0\},\{A_0,B_1\},\{A_1,B_0\},\{A_1,B_1\}\}$.
}
\end{figure}

\section{Convex cones}
\label{sec:cones}

In this section, we state several basic facts about closed convex
cones and their duals. Detailed background and proofs can be found in
\cite{aliprantis2007cones}. General text on convexity that also treat
cones are \cite{barvinok2002course, rockafellar1970convex}. All cones
that appear in this paper are closed and convex, so we will at times
drop the attributes.

A \emph{closed convex cone} $C$ is a subset of $\R^n$
\begin{enumerate}
	\item
	closed $\bar C = C$,
	\item
	convex, and
	\item
	scale-invariant: $\lambda C = C$ for every $\lambda\geq 0$.
\end{enumerate}
A simple example is given in
Figure~\ref{fig:acone}.

\begin{figure} [!t]
\centering
\includegraphics[width=0.35\textwidth]{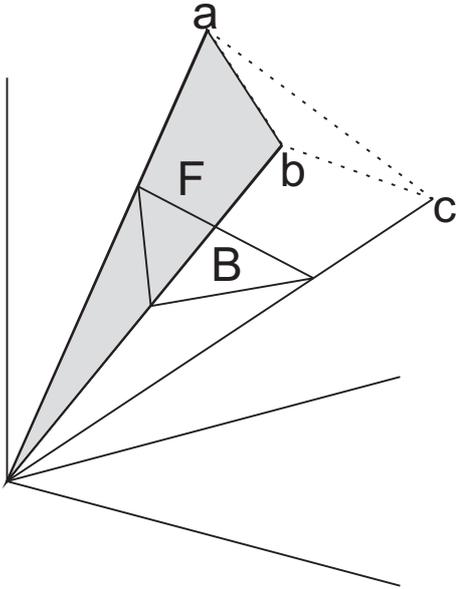}
\caption{
	\label{fig:acone}
	A closed convex cone in $\R^3$. The extremal rays are labeled by $a,
	b, c$, while $B$ designates a base. One of three facets is shaded
	and labeled $F$.
}
\end{figure}

The simplest types of cones are \emph{rays}, i.e.\ sets of the form
$\{ \lambda v \,|\, \lambda \geq 0\}$ for some vector
$v\in\mathbbm{R}^n$. Let $L\subset C$ be a ray contained in a closed
convex cone $C$. It is an \emph{extremal ray} if it cannot be written
as a non-trivial convex combination of elements in $C$, i.e.\ if
for all $x, y \in C$, whenever $\frac12(x+y)\in L$, we already have
that $x, y \in L$.

Under a technical assumption, closed convex cones are the convex hull
of their extremal rays. To state the assumption, we need to introduce
the notion of a \emph{base}. A base is a convex subset $B\subset C$ of
a convex cone $C$ such that $0\not\in B$ and every element $v\in C$ is
uniquely of the form $v=\lambda\,b$ with $\lambda\geq 0$ and $b\in B$.
Not every cone admits a base ($C=\R^n\subset \R^n$, e.g.\ does not).
However, cones which have a compact base are the convex hull of their
extremal rays
\cite[Chapter~9]{barvinok2002course}. This will be true for all cones
that we will deal with in this paper.

In this sense, it is sufficient to specify the extremal rays in order
to specify $C$. Thus, cones that have only finitely many extremal rays
are of particular interest. A cone has this property if and only if
it is the region in $\R^n$ specified by finitely many linear and
homogeneous inequalities \cite[Chapter 3.4]{aliprantis2007cones}.
Such cones are called \emph{polyhedral}. The (closure) of all
achievable entropy vectors is now known \emph{not} to be polyhedral
\cite{matus2007infinitely}. However, the cone $\Gamma_n$ is manifestly defined by
finitely many inequalities and hence polyhedral. The same is true for
all other cones that we will be working with.

There is a powerful notion of \emph{duality} for closed convex cones.
Let $C$ be such a cone. The \emph{dual cone} (also \emph{polar cone})
$C^*$ is the set of all homogeneous linear inequalities valid on $C$:
\begin{equation*}
	C^* = \{ f \,|\, \langle f, v \rangle \geq 0 \quad \forall\, v \in
	C\}.
\end{equation*}
In this language, the set of Shannon-type inequalities is just the
dual cone $\Gamma_n^*$ to $\Gamma_n$. The generating set in
(\ref{shannonineqs_basic})
are the extremal rays of $\Gamma_n^*$.
We will need the following properties of the duality operation:
\begin{enumerate}
	\item
	By the \emph{Bipolar Theorem}, $(C^*)^* = C$ for every closed convex
	cone $C$ \cite[Chapter 4]{barvinok2002course}. In particular, a cone is
	completely specified by its dual.

	\item
	Duality reverses inclusions
	\cite[Chapter 4]{barvinok2002course}:
	If $C, C'$ are closed convex cones and $C' \subset C$ then
	$C^* \subset (C')^*$.

	\item\label{item:contra}
	Dual cones transform ``contragradiently'': Let $C$ be a closed convex
	cone and $D$ a linear map. Then $C':=D(C)$ is again a convex cone and
	\begin{equation}\label{eqn:duality}
		D^T\big((C')^*\big)
		=	
		\operatorname{range} D^T \cap C^*
		\subset C^*,
	\end{equation}
	where $D^T$ is the adjoint of $D$.
\end{enumerate}

\begin{proof}[Proof of Property~\ref{item:contra}]
	Let $C$ be a closed convex cone, $D$ a linear map, and $C'=D(C)$. Then
	\begin{eqnarray*}
		D^T\big( (C')^* \big)
		&=&
		\{ D^T(f) \,|\, \langle f, v' \rangle \geq 0 \quad \forall\, v' \in C'\}
		\\
		&=&
		\{ D^T(f) \,|\, \langle f, D(v) \rangle \geq 0 \quad \forall\, v \in C\} \\
		&=&
		\{ D^T(f) \,|\, \langle D^T(f), v \rangle \geq 0 \quad \forall\, v \in C\} \\
		&=&
		\operatorname{range} D^T \cap C^*.
	\end{eqnarray*}
\end{proof}

\section{The correspondence between probabilistic and entropic inequalities}
\label{sec:backward_theorem}

In this section, we will present a simple geometric construction that
explains and generalizes the connection, observed above, between the
entropic $CHSH^E$ inequality and the usual $CHSH$ inequality. We will
find that the set of probability distributions for $n$ binary
experiments can be imbedded into the cone $\Gamma_n$ of set functions
fulfilling the polymatroidal axioms. Dually, it follows that every
linear inequality valid for $\Gamma_n$ can be turned into an
inequality valid for probability distributions. The linear map that
connects the two types of inequalities will turn out to send $CHSH^E$
to $CHSH$, thus providing a geometric explanation for the observed
coincidence. (Figure~\ref{fig:relations}) provides a high-level
roadmap through the succession of convex cones that appear in the
argument).

\begin{figure*} [!t]
\centering
\includegraphics[width=0.95\textwidth]{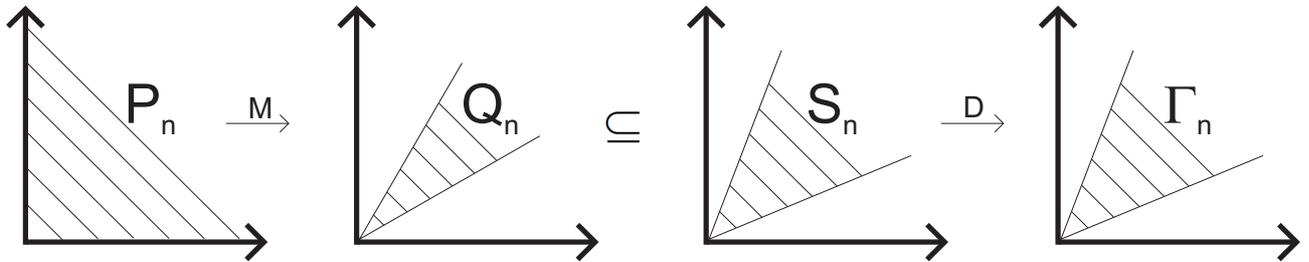}
\caption{
	\label{fig:relations}
	The various cones appearing in the argument that entropic
	inequalities can be mapped to Bell inequalities.
	We start with the positive orthant $P_n$ of $R_n$, which is the
	cone over the set of probability distributions. The M\"obius
	transform $M$ sends it linearly and bijectively to a cone $Q_n$
	($Q_n$, and all further cones that appear later,
	happen to be a sub-cone of the positive orthant. Thus, in this way,
	our
	two-dimensional sketch is faithful). The set of M\"obius-transformed
	distributions is a sub-cone of $S_n$, a cone which fulfills a set of
	``inverted'' polymatroidal axioms. The latter cone can be imbedded
	linearly into the Shannon cone $\Gamma_n$. Using cone duality
	(\ref{eqn:duality}), we can invert the chain above and imbed the
	dual cone $\Gamma_n^*$ into $Q_n^*$. That yields the main claim of
	this section. Note that the initial M\"obius transform is not
	strictly necessary to arrive at $D^T(\Gamma^*) \subset Q_n^*$. We have
	stated it primarily to clarify the geometric nature of $Q_n$ (i.e.\
	as an orthant, up to a linear isomorphism).
}
\end{figure*}

We start by considering various ways of representing the probability distribution of $n$ binary random variables $X_1, \dots, X_n$.Most naturally, the distribution is given by a function
\begin{equation*}
	p: \{ 0, 1\}^{\times n} \to [0,1]
\end{equation*}
on binary strings of length $n$ with the interpretation that
\begin{equation*}
	p(x)=\Pr[X_1 = x_1, \dots, X_n = x_n].
\end{equation*}
Let $x$ be an $n$-bit string. The string is obviously characterized
by the set $A\in 2^{[n]}$ of the positions where it equals $0$. Hence
we can equivalently consider $p$ as a function on the set of subsets
of $[n]$:
\begin{eqnarray*}
	p: 2^{[n]} &\to& [0,1], \\
	p(A) &=& \Pr[X_A = 0 \,\wedge \, X_{A^C}=1],
\end{eqnarray*}
where, again, $X_A$ are those components of the random vector $X$ whose indices appear in
the set $A$. With this convention, $p$ can be seen as an element of
the  real vector space $R_n$ over the powerset of $[n]$. More
precisely, it is an element of the non-negative orthant of $R_n$, and
every element of that orthant corresponds to a (not necessarily
normalized) distribution. We denote the non-negative orthant of $R_n$
by
\begin{equation*}
	P_n := \left\{ p \in R_n \,|\, p_A \geq 0 \quad \forall\,A\in 2^{[n]}
	\right\}.
\end{equation*}

The reason we found it necessary to elaborate on this rather straight-forward correspondence is that the $CHSH$ inequality
(\ref{CHSH}) is given in terms of a different parametrization of probability distributions, which we can now explicitly connect to the
standard one. Indeed, the quantities appearing in (\ref{CHSH}) are
these:
\begin{eqnarray}
	q: 2^{[n]} &\to& [0,1], \nonumber \\
	q(A) &=& \Pr[X_A = 0 ]
	= \sum_{B, A\subset B} p(B). \label{eqn:moebiusforward}
\end{eqnarray}
Equation (\ref{eqn:moebiusforward}) defines a linear map $M: R_n \to
R_n$ such that $q = M p$. A priori, it is not clear that $M$ is
invertible, i.e.\ that one can specify a distribution in terms of the
``$q$-vector'' above. However, that turns out to be true.
In essence, the relation is given by the M\"obius inversion formula
\cite[Chapter 6]{brualdi2012introductory}.

\begin{lemma}
	The linear map $M$ defined by
	\begin{eqnarray*}
		q(A):=(Mq)(A) = \sum_{B, A \subset B} p(B)
	\end{eqnarray*}
	is invertible. Its inverse is given by
	\begin{equation*}
		p(A) := (M^{-1}p)(A) =
		\sum_{B, A \subset B^C} (-1)^{|A|-|B|} q(B^C).
	\end{equation*}
\end{lemma}

The superscript $C$ stands, of course, for the set complement within
$[n]$.

\begin{proof}
	A few manipulations bring the problem into a standard form of the
	M\"obius transformation (we use the notions of
	\cite[Chapter 6.6]{brualdi2012introductory}).
	Using (\ref{eqn:moebiusforward}) and repeatedly re-labeling the sets one sums
	over:
	\begin{eqnarray*}
		q(A^C)
		&=& \sum_{B, A^C \subset B} p(B) \\
		&=& \sum_{B^C, B^C \subset A} p(B) \\
		&=& \sum_{B, B \subset A} p(B^C).
	\end{eqnarray*}
	Thus \cite[(6.10), (6.11)]{brualdi2012introductory} apply with
	$G(A)=q(A^C), F(B)=p(B^C)$. In particular, we find
	\begin{eqnarray*}
		p(A^C) &=& \sum_{B\subset A} (-1)^{|A^C|-|B|} q(B^C)
	\end{eqnarray*}
	which is the stated relation, up to an additional
	re-parameterization of $A\to A^C$.
\end{proof}

The set of non-negative distribution in $q$-representation is thus the
M\"obius transform of the non-negative orthant. We denote it by
\begin{equation*}
	Q_n := M(P_n) = \{ M p, \,|\, p\in P_n\}.
\end{equation*}
The significance of $Q_n$ is that its elements fulfill a set of
``inverted'' polymatroid axioms. In order to state this precisely, we
have to introduced yet another (and final!) cone.

\begin{definitionlemma}\label{deflem:sn}
	The cone $S_n$ is the set of vectors $s\in R_n$ that are
	\begin{enumerate}
		\item
		\emph{non-negative}
		\begin{equation*}
			s_A\geq 0 \quad \forall\,A\in 2^{[n]},
		\end{equation*}

		\item
		\emph{decreasing}
		\begin{equation*}
			s_B \geq s_A \text{ for $B \subseteq A$ },
		\end{equation*}

		\item
		\emph{super-modular}
		\begin{equation*}
			h_{A\cup B} + h_{A \cap B} \geq h_{A} + h_{B}.
		\end{equation*}
	\end{enumerate}

	It holds that $Q_n \subset S_n$.
\end{definitionlemma}

\begin{proof}
	Let $q=Mp$ be the M\"obius transform of a probability distribution.
	We will verify the properties 1.\ -- 3.\ in turn. Since they are
	obviously invariant under re-scaling by a positive number, this
	suffices to conclude $Q_n\subset S_n$.

	Positivity follows directly from the definition of a probability.
	Property~(2) is likewise a straight-forward consequence of
	(\ref{eqn:moebiusforward}): If $B\subset A$, then the probability
	that all $X_B$ are simultaneously zero is certainly larger than or
	equal to the probability  that even all $X_A$ are equal to zero.

	As for super-modularity: For any event $E$, let $\delta(E)$ be the
	``indicator function'' that takes the value $1$ if $E$ occurs and
	$0$ else. The inequality
	\begin{equation*}
		\delta( X_{A\cup B} = 0 ) +
		\delta(X_{A\cap B} = 0 )
		\geq
		\delta(X_A=0)+\delta(X_B = 0)
	\end{equation*}
	holds with probability one. Indeed, as soon as one of the terms on
	the right hand side (r.h.s.) is one, $\delta(X_{A\cap B})$ will also
	be one; if both terms on the r.h.s.\ are one, then so are both
	summands on the l.h.s.
	Super-modularity now follows from taking expectations on both sides.
\end{proof}

The remainder of the argument will proceed as follows: We observe that there is a linear map $D$ that sends $S_n$ onto $\Gamma_n$. It then
follows from elementary convex geometry (Section~\ref{sec:cones}) that the dual map $D^T$ sends linear inequalities valid on $\Gamma_n$ (i.e.\ Shannon-type
inequalities) to linear inequalities valid on $S_n$. Since $Q_n\subset S_n$, the inequalities also hold for M\"obius-transformed probability
distributions. The following statements make this precise.

\begin{lemma}\label{backward}
	Let $D: R_n \to R_n$ be defined by
	\begin{equation*}
		(D s)_A = s_{\emptyset} - s_A.
	\end{equation*}
	Then $D(S_n)\subset \Gamma_n$.
\end{lemma}

\begin{proof}
	Let $A\subseteq B$ and $s\in S_{n}$, then the inequality
	\begin{equation}\label{increasing}
	 	D(s)_{A}=s_{\emptyset}-s_{A}\leq s_{\emptyset}-s_{B}=D(s)_{B}
	\end{equation}
	follows from the fact that vectors $s$ in $S_n$ have decreasing
	components.  The inequality
	\begin{eqnarray}\label{submodular}
		&& D(s)_{A}+D(s)_{B}  \\
		&=& 2s_{\emptyset}-[s_{A}+s_{B}] \nonumber \\
		&\geq& 2s_{\emptyset}-[s_{A\cup B}+s_{A\cap B}]  \nonumber \\
		&=&  D(s)_{A\cup B}+D(s)_{A\cap B}
	\end{eqnarray}
	follows from super-modularity of $s\in S_{n}$. Next,
	\begin{equation}\label{positive}
		D(s)_{A}=s_{\emptyset}-s_{A}\geq 0
	\end{equation}
	follows from the fact that $s$ is decreasing. Finally
	\begin{equation}\label{emptyset}
		D(\emptyset)_{A}=s_{\emptyset}-s_{\emptyset}= 0.
	\end{equation}

	Inequalities \eqref{increasing}, \eqref{submodular},
	\eqref{positive} and \eqref{emptyset} show, respectively, that
	$D(s)$ is monotonously increasing, sub-modular, non-negative and
	that its $\emptyset$-component is zero.
	These are the defining properties of the Shannon-cone. Hence
	$D(s)\in \Gamma_{n}$.
\end{proof}

We thus find that any Shannon-type inequality can be mapped to an
inequality valid for any M\"obius-transformed probability
distribution:

\begin{corollary}
	Let $\mathcal{M}$ be a marginal scenario and let
	$f\in(\Gamma^{\mathcal{M}})^*$
	be a Shannon-type inequality.
	Then
	\begin{equation*}
		D^{T}(f)\in Q_{n}^{*},
	\end{equation*}
	i.e.\ $D^T(f)$ holds for
	M\"obius-transformed probability distributions.
\end{corollary}

\begin{proof}
	We combine properties 2.\ and 3.\ of cone duality as stated in  Section~\ref{sec:cones} with
	Lemmas~\ref{deflem:sn}, \ref{backward} to obtain
	\begin{equation*}
		D^{T}(\Gamma_{n}^{*})\subseteq S_{n}^{*}\subset Q_{n}^{*}.
	\end{equation*}
	Since $(\Gamma^{\mathcal{M}})^* \subset \Gamma_n^*$ for any marginal
	scenario
	$\mathcal{M}$, we are done.
\end{proof}

\subsection{Discussion}

The space $R_n$ is equipped with a basis $e_A, A\in 2^{[n]}$ labeled by
subsets of $\{1, \dots, n\}$. If one orders the basis in any way such
that $e_\emptyset$ is the first element, then the linear map $D$ takes
the form
\begin{equation}
 D=\begin{pmatrix}
0 & 0 & 0 & 0 &  0 & \dots \\
1 & -1 & 0 & 0 & 0 & \dots \\
1 & 0 & -1 & 0 & 0 & \dots \\
1 & 0 & 0 & -1 & 0 & \dots \\
\vdots &&&& \ddots
\end{pmatrix}
\end{equation}
and its transpose is
\begin{equation}
 D^T=\begin{pmatrix}
0 & 1 & 1 & 1 &  1 & \dots \\
0 & -1 & 0 & 0 & 0 & \dots \\
0 & 0 & -1 & 0 & 0 & \dots \\
0 & 0 & 0 & -1 & 0 & \dots \\
\vdots &&&& \ddots
\end{pmatrix}
\end{equation}
Written as a vector, the
entropic $CHSH^E$ inequality (\ref{ECHSH}) reads
\begin{equation*}
	f
	=
	e_{\{A_0, B_0\}}+
	e_{\{A_0, B_1\}}+
	e_{\{A_1, B_0\}}-
	e_{\{A_1, B_1\}}-
	e_{\{A_0\}}-
	e_{\{B_0\}}.
\end{equation*}
Because the coefficients sum to zero, $(D^T f)_\emptyset=0$ and hence
$D^T f = - f$, which is the vector representing the ordinary $CHSH$
inequality (\ref{CHSH}). We have thus indeed geometrically explained
the coincidence observed initially.

We remark that the inclusion $D^T(\Gamma^*_n) \subset Q_n^*$ is not
tight in general. I.e.\ it is not the case that all inequalities for
$Q_{n}$ can be obtained from those of $\Gamma_n$. Geometrically, this
would be surprising, as $Q_n$ is just an orthant, while $\Gamma_n$
seems to be a more complicated geometrical object. It is indeed simple
to find  explicit counter-examples:

Consider a specific inequality, for instance the Mermin inequality for
tripartite correlations \cite{Mermin1990}, it is possible to gain a
better intuition.  This is an example of an inequality that is valid
on q functions but can not be translated into an entropic inequality.
The reason is that for its derivation $32$ independent inequalities
(arising from positivity of some probability distribution) are needed;
one being the positivity of
$q_{\Omega}\equiv q_{\{A_{0},A_{1},B_{0},B_{1},C_{0},C_{1}\}}$, three correspond to the
decreasing property, six are super-modularities and there are other
$22$ inequalities that cannot be translated into Shannon type
inequalities.

\section{Collins-Gisin entropic inequalities with and without bounded shared randomness}
\label{sec:Imm}
In this section we derive an entropic version of the Collins-Gisin (CG) inequalities \cite{Collins2004}, concerning a bipartite scenario where each party, say Alice and Bob, can choose between $m$ measurement settings each.
We further derive a different version of these inequalities that take into account a bounded amount of shared randomness between the parties.

The CG inequalities are typically written in the form $ I_{mm22} \leq 0$, where for $m=2$ this corresponds to the CHSH inequality \cite{Clauser1969}. The notation of the inequality $ I_{mm22}$ stresses that each party has
access to $m$ possible measurement settings with $2$ outcomes each. For $m=3$ it has been shown that these inequalities are useful since they can detect the nonlocality of states that can not be detected by the CHSH
inequality \cite{Collins2004}. Moreover, as shown in \cite{Avis2005}, the $I_{mm22}$ are tight Bell inequalities, that is, they correspond to facets of the local polytope.

The $I_{mm22}$ inequality can be written compactly using the following matrix notation
\begin{equation}
\label{matrix_notation}
\left(
\begin{tabular}{c||cccc}
   & $q_{A_0}$ & $q_{A_1}$ & $\cdots$ & $q_{A_{m-1}}$ \\
  \hline \hline
  $q_{B_0}$ & $q_{A_0B_0}$ & $q_{A_1B_0}$ & $\cdots$ & $q_{A_{m-1}B_0}$ \\
  $q_{B_1}$ & $q_{A_0B_1}$ & $q_{A_1B_1}$ & $\cdots$ & $q_{A_{m-1}B_1}$ \\
  $\vdots$ & $\vdots$ & $\vdots$ & $\ddots$ & $\cdots$ \\
  $q_{B_{m-1}}$ & $q_{A_0B_{m-1}}$ & $q_{A_1B_{m-1}}$ & $\cdots$ & $q_{A_{m-1}B_{m-1}}$ \\
\end{tabular}
\right).
\end{equation}
Using this matrix notation the $I_{mm22}$ inequality can be written as
\begin{equation}
\label{Imm_prob}
I_{mm22} = \left(
\begin{tabular}{c||cccccc}
   & -1 & 0 & 0 & 0 & 0 & 0\\
  \hline \hline
  -(m-1) & 1 & 1 & $\cdots$ & 1 & 1 & 1\\
  -(m-2) & 1 & 1 & $\cdots$ & 1 & 1& -1 \\
  -(m-3) & 1 & 1 & $\cdots$ & 1 & -1 & 0\\
  $\vdots$ & $\vdots$ &  $\vdots$ & $\vdots$ & $\ddots$ & $\vdots$ & $\vdots$\\
  0 & 1 & -1 & 0 &  $\cdots$ & 0 & 0\\
\end{tabular}
\right) \leq 0.
\end{equation}
It is important to stress the difference between the way one proves the validity of an entropic inequality and the validity of a probability inequality. In general to prove that a probability inequality is valid, one uses the
information about the extreme points of the local polytope, that is, all the deterministic functions assigning values to the outcomes. In turn, as stressed before, little is known about the extremal rays of the Shannon-type
entropic cone (apart from simple cases ~\cite{Kashiwabara1996,Shapley72,Studeny2000}). In the absence of information about the extremal rays of the entropic cone, the only way we can prove that the entropic inequality is
valid is to use the linear programm approach of Yeung \cite{Yeung2008}. If the extremal rays are known, a very similar approach to the one used in correlation polytopes \cite{Pitowsky1989} can also be used in the entropic
case (See Appendix \ref{find_facets} for further details).

We start considering the case $m=3$. From Theorem 6 and the corresponding translation rule $H_{A_iB_j} \rightarrow -q_{A_iB_j}$ one could expected that the entropic analogous of $I_{3322}$, that we label as $I^{E}_{33}$,
could be given by
\begin{align}
\label{I3322}
I^{E}_{33}=&H_{A_0B_0} +H_{A_1B_0} +H_{A_2B_0} + H_{A_0B_1} \\ \nonumber
&  +H_{A_1B_1} -H_{A_2B_1} +H_{A_0B_2} - H_{A_1B_2} \\ \nonumber
& -2H_{B_0} -H_{B_1}-H_{A_0} \geq 0
\end{align}
This is indeed the case as this inequality can be obtained by the combination of the following basic inequalities
\begin{align}
\label{pentagonal}
 &H_{A_0A_1B_0}        +H_{A_0A_1B_1}  \geq H_{A_0A_1B_0B_1} + H_{A_0A_1} \\
 &H_{A_0B_0}        +H_{A_1B_0}  \geq H_{A_0A_1B_0} + H_{B_0} \\
 &H_{A_0B_1}        +H_{A_1B_1}  \geq H_{A_0A_1B_1} + H_{B_1}\\
 &H_{A_2B_0}+ H_{B_0B_1} \geq H_{A_2B_0B_1}+H_{B_0}\\
 &H_{A_0B_2}+ H_{A_0A_1} \geq H_{A_0A_1B_2}+H_{A_0}
\end{align}
together with the following monotonicity inequalities $H_{A_2B_0B_1} \geq H_{A_2B_1}$, $H_{A_0A_1B_2} \geq H_{A_1B_2}$ and $H_{A_0A_1B_0B_1}\geq H_{B_0B_1} $ (remember that all monotonicity inequalities can be obtained by
the basic inequalities). The notation of the inequality $ I^{E}_{mm}$ stresses that each party $A$ and $B$ has access to $m$ possible measurement settings with any number of possible outcomes, in contrast to the $I_{mm22}$
inequalities that are only valid for dichotomic observables. As discussed in the introduction this outcome size independence is an advantage of the entropic inequalities over the probabilistic ones.

In the Appendix \ref{Imm_app} it is proven, proceeding with a similar FM elimination as the one sketched above, that the CG inequalities are valid for entropies if one simply applies the transformation rule $H_{A_iB_j}
\rightarrow -q_{A_iB_j}$, that is,
\begin{equation}
\label{Imm_ent}
I^{E}_{mm} = \left(
\begin{tabular}{c||cccccc}
   & 1 & 0 & 0 & 0 & 0 & 0\\
  \hline \hline
  (m-1) & -1 & -1 & $\cdots$ & -1 & -1 & -1\\
  (m-2) & -1 & -1 & $\cdots$ & -1 & -1& 1 \\
  (m-3) & -1 & -1 & $\cdots$ & -1 & 1 & 0\\
  $\vdots$ & $\vdots$ &  $\vdots$ & $\vdots$ & $\ddots$ & $\vdots$ & $\vdots$\\
  0 & -1 & 1 & 0 &  $\cdots$ & 0 & 0\\
\end{tabular}
\right) \leq 0,
\end{equation}
were we have used a similar notation to the one in \eqnref{matrix_notation}. From Theorem 6, this also implies that the $I_{mm22}$ inequalities \eqref{Imm_prob} can be derived relying exclusively on the inverse polymatroidal
axioms.

Given the inequality \eqnref{Imm_ent} the first question one needs to answer is if it is able to witness nonlocal correlations. For the usual CG inequality \eqref{Imm_prob} the maximal violation is achieved by the nonlocal
non-signalling distribution
\begin{equation}
\label{pm}
p_m\left(  a,b | x,y \right)  =\left\{
\begin{array}{ll}
1/2 & \text{, } a\oplus b=1 \text{, } x+y=m   \\
1/2 & \text{, } a\oplus b=0 \text{, } x+y \neq m   \\
0 & \text{, otherwise}%
\end{array}
\right. ,
\end{equation}
that can be understood as a generalization of the paradigmatic PR-box \cite{Popescu1994} for $m$ measurement settings. If we directly compute the value of $I^{E}_{mm}$ for the distribution $p_m$ we find no violations. This
is no surprise since entropies are unable to distinguish between correlations and anti-correlations; for example, $p_m$ is entropically equivalent to the classically correlated distribution
\begin{equation}
\label{pc}
p_c\left(  a,b | x,y \right)  =\left\{
\begin{array}{ll}
1/2 & \text{, } a\oplus b=0   \\
0 & \text{, otherwise}%
\end{array}
\right. .
\end{equation}
In order to find violations of the entropic inequalities one needs to find a way of entropically distinguishing correlations from anti-correlations. As shown in \cite{Chaves2013} one way to do that is to make use of shared
randomness between the parties. Consider two distributions $p_c$ and $p_a$ that have, respectively, correlated outputs ($ a\oplus b =0$) and anti-correlated outputs ($ a\oplus b =1$), whatever the inputs. Entropically both
distributions are indistinguishable but if we allow the parties to make use of some extra shared randomness then we can tell apart both distributions. For example, mixing with equal probabilities the distributions with an
independent copy of $p_c$, we see that $p_c$ remains unchanged while $p_a$ is turned into a uncorrelated distribution. Similarly if we mix $\frac{1}{2} p_m + \frac{1}{2} p_c$ we see that $I^{E}_{mm}=m-1$, a violation of the
entropic inequality that can be proven to be optimal, that is, in some sense (allowing the use of shared randomness) the maximally nonlocal probability distribution is also the maximally entropically nonlocal.

To prove the maximal violation of $I^{E}_{mm} \leq 0$ we first consider the maximum algebraic value that the operator $I^{E}_{mm}$ can achieve, that turns out to be the same as the one obtained with the distribution
$\frac{1}{2} p_m + \frac{1}{2} p_c$. In order to understand the maximal violations $I^{E}_{mm}$ let us rewrite it in terms of mutual informations
\begin{equation}
\label{I_I}
I^{E}_{mm} = \left(
\begin{tabular}{c||cccccc}
   & -(m-2) & -(m-3) & -(m-4) & $\cdots$ & -1 & 0\\
  \hline \hline
  -1 & 1 & 1 & $\cdots$ & 1 & 1 & 1\\
  0 & 1 & 1 & $\cdots$ & 1 & 1& -1 \\
  0 & 1 & 1 & $\cdots$ & 1 & -1 & 0\\
  $\vdots$ & $\vdots$ &  $\vdots$ & $\vdots$ & $\ddots$ & $\vdots$ & $\vdots$\\
  0 & 1 & -1 & 0 &  $\cdots$ & 0 & 0\\
\end{tabular}
\right) \leq 0 .
\end{equation}
where we have used the matrix notation
\begin{equation}
\label{matrix_notation2}
\left(
\begin{tabular}{c||cccc}
   & $H_{A_0}$ & $H_{A_1}$ & $\cdots$ & $H_{A_{m-1}}$ \\
  \hline \hline
  $H_{B_0}$ & $I_{A_0:B_0}$ & $I_{A_1:B_0}$ & $\cdots$ & $I_{A_{m-1}:B_0}$ \\
  $H_{B_1}$ & $I_{A_0:B_1}$ & $I_{A_1:B_1}$ & $\cdots$ & $I_{A_{m-1}:B_1}$ \\
  $\vdots$ & $\vdots$ & $\vdots$ & $\ddots$ & $\cdots$ \\
  $H_{B_{m-1}}$ & $I_{A_0:B_{m-1}}$ & $I_{A_1:B_{m-1}}$ & $\cdots$ & $I_{A_{m-1}:B_{m-1}}$ \\
\end{tabular}
\right).
\end{equation}
Using that $I_{A_iB_j}-H_{A_i} \leq 0$ and $I_{A_0B_0}-H_{B_0} \leq 0$  we see that the maximum violation is given by $I^{E}_{mm}= -\textstyle{\sum}_{i=0,\dots,m-2} I_{A_{1+i}:B_{m-1-i}}+\textstyle{\sum}_{i=0,\dots,m-2}
I_{A_{1+i}:B_{m-2-i}} \leq \textstyle{\sum}_{i=0,\dots,m-2} H_{A_{1+i}} $. For dichotomic observables it turns out that the maximal violation is given by $I^{E}_{mm} \leq m-1$.

\subsection{Entropic CG inequality with bounded shared randomness}
\label{subsec:boundedlambda}

With the locality and realism assumption any correlation displayed between $A$ and $B$ can only occur through the hidden variable $\lambda$. The variable $\lambda$ is the common ancestor to all the
observable quantities. That means that any correlation shown between $A$ and $B$ must be screened off if we know the actual value of $\lambda$. Mathematically this corresponds to say that $I_{A_0A_1:B_0B_1|\lambda}=0$
(similarly to all subsets, for example, $I_{A_0:B_0|\lambda}=0$) or in other terms $H_{A_0A_1B_0B_1|\lambda}=H_{A_0A_1|\lambda}+H_{B_0B_1|\lambda}$ that can be rewritten as  $H_{A_0A_1B_0B_1\lambda} + H_{\lambda}
=H_{A_0A_1\lambda}+H_{B_0B_1\lambda}$. Remember that the mutual information between two variables can be expressed in terms of Shannon entropies as $I_{A:B}=H_{A}+H_{B}-H_{AB}$ and similarly
$I_{A:B|\lambda}=H_{A\lambda}+H_{B\lambda}-H_{AB\lambda}-H_{\lambda}$. Note that we allow the parties to have access to local randomness, that is, $H_{A_0A_1|\lambda}$ (similarly to B) not necessarily is equal to $0$. Our
aim is to bound the entropy of the hidden variable to be $H_{\lambda} \leq \mathcal{C}$.

Such a restriction fits naturally in the entropic approach to marginal models, since the considerations about finite shared randomness are equivalent to extra linear constraints that still define an entropic cone. In practice we start considering all the polymatroidal axioms describing the cone $\Gamma_{n+1}$, corresponding to all $n$ variables of the marginal scenario plus the
hidden variable $\lambda$. We add to this set of basic inequalities the ones that contain the information about the causal structure of the experiment, that is, saying that $\lambda$ is the only common ancestor to all the
space-like separated variables and also the inequality bounding the entropy of the hidden variable. Formally, this means we add to the set of basic inequalities, the following inequalities
\begin{eqnarray}
\label{lambdadet}
& I_{A:B|\lambda}=0 \\
\label{boundlambda}
& H_{\lambda} \leq \mathcal{C}
\end{eqnarray}
where $A=(A_{0},...,A_{m-1})$ and $B=(B_{0},...,B_{m-1})$.

The first step in the FM elimination is to eliminate the hidden variable $\lambda$. Note that since \eqref{boundlambda} is the only inequality that depends on $\mathcal{C}$, after the FM elimination of the variable $\lambda$
any non-trivial inequality depending on the amount of shared randomness should appear as the sum of this inequality with some of the other inequalities. To begin with, we now prove that
\begin{equation}
\label{boundedlambda}
H_{A}+H_{B}\leq H_{AB}+\mathcal{C},
\end{equation}
where again $A=(A_{0},...,A_{m-1})$ and $B=(B_{0},...,B_{m-1})$.

To obtain \eqref{boundedlambda} we add the independence condition $I_{A:B|\lambda}=0$ with one basic submodularity inequality, one basic monotonicity inequality and the bound on $H(\lambda)$:
\begin{align}
 &H_{A,\lambda}+H_{B,\lambda}=H_{A,B,\lambda}+H_{\lambda} \\
 &H_{A,B,\lambda}+H_{A}\leq H_{A,B}+H_{A,\lambda} \\
 &H_{B}\leq H_{B,\lambda} \\
 &H_{\lambda}\leq \mathcal{C}
\end{align}
Note that in the limit that $\mathcal{C}=0$, since $ H_{B} +H_{A} \geq H_{A B}$ this implies that $ H_{B} +H_{A} = H_{AB}$, that is, no correlations between A and B are possible, as one should expect. We have checked
computationally that \eqref{boundedlambda} is the only extra facet inequality to the usual basic set one gets after eliminating $\lambda$ for the CHSH scenario. We believe this is still the case for scenarios with more
measurement settings but we do not have a formal proof of that.

Note that all the terms appearing in \eqref{boundedlambda} involve non-observable quantities and should then be eliminated. Our approach here is to add basic inequalities in such a way that we eliminate all the
non-observable quantities. Combining the following basic inequalities
\begin{align}
\label{ECHSHsharedcombine}
 &H_{A_1B_1} +H_{A_0A_1B_0B_1}         \leq H_{A_1B_0B_1}        +H_{A_0A_1B_1} \\
 &H_{A_0}       +H_{A_0A_1B_1}         \leq H_{A_0A_1} + H_{A_0B_1}          \\
 &H_{B_0}     +H_{A_1B_0B_1}   \leq  H_{B_0B_1}  +H_{A_1B_0}
\end{align}
with inequality \eqref{boundedlambda} we obtain
\begin{align}
\label{ECHSHshared}
\nonumber
BI^{E}_{22} & =-H_{A_0B_1}-H_{A_1B_0}+H_{A_1B_1}+ H_{A_0} +H_{B_0}  \leq \mathcal{C} \\
& =I_{A_0B_1}+I_{A_1B_0}-I_{A_1B_1} - H_{A_0}  \leq \mathcal{C}
\end{align}
that one can regard as the entropic CHSH with bounded shared randomness.

For general $m$, as proven in the Appendix, the following inequality can be regarded as the entropic CG inequality with bounded shared randomness,
\begin{equation}
\label{BI_H}
BI^{E}_{mm} = \left(
\begin{tabular}{c||cccccc}
   & 1 & 0 & 0 & $\cdots$ & 0 & 0\\
  \hline \hline
  1 & 0 & 0 & $\cdots$ & 0 & 0 & -1\\
  m-2 & -1 & -1 & $\cdots$ & -1 & -1& 1 \\
  m-3 & -1 & -1 & $\cdots$ & -1 & 1 & 0\\
  $\vdots$ & $\vdots$ &  $\vdots$ & $\vdots$ & $\ddots$ & $\vdots$ & $\vdots$\\
  0 & -1 & 1 & 0 &  $\cdots$ & 0 & 0\\
\end{tabular}
\right) \leq \mathcal{C} .
\end{equation}
where once more we have used a matrix notation similar to the one in \eqnref{matrix_notation}. In terms of the mutual information, the inequality can be written as (using the matrix notation \eqref{matrix_notation2})
\begin{equation}
\label{BI_I}
BI^{E}_{mm} = \left(
\begin{tabular}{c||cccccc}
   & -(m-2) & -(m-3) & -(m-4) & $\cdots$ & 0 & 0\\
  \hline \hline
  0 & 0 & 0 & $\cdots$ & 0 & 0 & 1\\
  0 & 1 & 1 & $\cdots$ & 1 & 1& -1 \\
  0 & 1 & 1 & $\cdots$ & 1 & -1 & 0\\
  $\vdots$ & $\vdots$ &  $\vdots$ & $\vdots$ & $\ddots$ & $\vdots$ & $\vdots$\\
  0 & 1 & -1 & 0 &  $\cdots$ & 0 & 0\\
\end{tabular}
\right) \leq \mathcal{C} .
\end{equation}

In order to understand the violation of the $BI^{E}_{mm}$ and its relation to $I^{E}_{mm}$ we first note that $BI^{E}_{m,m}=I^{E}_{m-1,m-1}+I_{A_{m-1},B_0}-I_{A_{m-1},B_1}+H_{B_0} \leq H_{\lambda}$. For any local
distribution $I^{E}_{m-1,m-1} \leq 0$ and we have that $I_{A_{m-1},B_0}-I_{A_{m-1},B_1}+H_{B_0} \leq H_{B_0}+\min{\left(H_{A_{m-1}},H_{B_0}\right)} $, that means that any local distribution is bounded by $BI^{E}_{mm}\leq
H_{B_0}+\min{\left(H_{A_{m-1}},H_{B_0}\right)}$. Consider for instance dichotomic observables such that for local distributions $BI^{E}_{mm}\leq 2$. That means that independently of how many measurement settings one employs
the inequality will be saturated resorting to not more than only two bits of shared randomness.

\section{Multipartite Scenarios}
\label{sec:Multi}
We start considering the simplest multipartite scenario, consisting of $3$ parties with $2$ measurement settings each. In terms of the correlation polytope it is known that there are $46$ different classes of
inequalities \cite{Sliwa2003}. As we discuss in Sec. \ref{sec:computational_results} the FM elimination method to obtain the entropic inequalities bounding the marginal scenario is too demanding and we were not able to
finish the computation.

To circumvent this limitation we proceed to derive a non-trivial inequality using the chain rule for entropies, a similar approach originally employed to derive the entropic CHSH inequality \cite{Braunstein1988}. Remember
that a marginal model in accordance with a LHV description assures the existence of the joint full probability distribution $p(a_{x=1},b_{y=1},c_{z=1},a_{x=0},b_{y=0},c_{z=0})$, with $x$, $y$, and $z$ describing the measurement
choices, for example $x=0$ corresponds to Alice measuring the observable $A_0$. The existence of the joint full distribution in turns imply the existence of the joint full entropy $H(A_1,B_1,C_1,A_0,B_0,C_0)$. Using the
chain rule for the entropies we have that
\begin{eqnarray}
\nonumber
& H_{A_1B_1C_1A_0B_0C_0}=H_{A_1|B_1C_1A_0B_0C_0}+H_{B_1|C_1A_0B_0C_0} \\
& + H_{C_1|A_0B_0C_0}+H_{A_0|B_0C_0}+H_{B_0|C_0}+H_{C_0},
\end{eqnarray}
that in turns implies that
\begin{eqnarray}
\label{M3}
\nonumber
& M_{3}= H_{A_1B_1C_1}-H_{A_1B_0C_0}-H_{A_0B_1C_0} - H_{A_0B_0C_1} \\
& -H_{A_0B_0C_0}+H_{A_0B_0}+H_{A_0C_0}+H_{B_0C_0} \leq 0 ,
\end{eqnarray}
where we have simply used the monotonicity of the Shannon entropy and the fact that conditioning on a variable cannot increase the entropy, that is, $H_{A} \leq H_{AB}$ and $H_{A|BC} \leq H_{A|C}$. Note that the chain rule
for entropies and the two other used properties aforementioned are Shannon type relations and as so the inequality $M_{3} \leq 0$ can also be derived from the basic set of inequalities \eqref{shannonineqs_basic}. Given the
inequality \eqnref{M3} the first question one needs to answer is if the inequality is able to detect genuine tripartite nonlocal correlations. To show that we consider the two kinds of genuine tripartite nonlocal correlations
introduced in \cite{Barrett2005}:
\begin{equation}
\label{p1}
p_1\left(  a,b,c | x,y,z \right)  =\left\{
\begin{array}{ll}
1/4 & \text{, } a\oplus b \oplus c =xyz\\
0 & \text{, otherwise}%
\end{array}
\right. ,
\end{equation}
and
\begin{equation}
p_2\left(  a,b,c | x,y,z \right)  =\left\{
\begin{array}{ll}
1/4 & \text{, } a\oplus b \oplus c =xy\oplus xz \oplus yz \\
0 & \text{, otherwise}%
\end{array}
\right. .
\end{equation}
If we compute the value $M_{3}$ for these distributions we find no violations. As mentioned before, this comes as no surprise since entropies are unable to distinguish between correlations and anti-correlations; for example,
both distributions $p_1$ and $p_2$ are entropically equivalent to the classically correlated distribution
\begin{equation}
p_c\left(  a,b,c | x,y,z \right)  =\left\{
\begin{array}{ll}
1/4 & \text{, } a\oplus b \oplus c =0\\
0 & \text{, otherwise}%
\end{array}
\right. .
\end{equation}
As discussed before one way to make the distinction between correlation and anti-correlation from the entropic perspective is to use classical shared randomness. If we just mix the distributions $p_1$ and $p_2$ with $p_c$,
for example equally mixing them with the same probability of $1/2$, one can straightforwardly compute the value of the operator to be $M_{3}=1$ for both distributions, a violation of the inequality \eqnref{M3} that therefore
witnesses the non-local behaviour of the distributions.

A nice feature of the inequality \eqnref{M3} is that it can be easily generalized for any number of parties $N$. Once more, just making use of the chain rule, the monotonicity of the Shannon entropy and the fact that
conditioning on a variable cannot increase the entropy we arrive at
\begin{eqnarray}
\label{Mn}
\nonumber
& M_{n}= H_{X^1_1 \dots X^N_1}-H_{X^1_0 \dots X^N_0} - P(H_{X^1_1 X^2_0 \dots X^N_0}) \\
& + P(H_{X^1_0 \dots X^{N-1}_0}) \leq 0,
\end{eqnarray}
where now we have used the notation $X^j_i$ to label the i-th observable of the j-th party with $i=\left\{ 0,1 \right\}$ and $j=1, \dots, N$. The operator P stands for all the different permutations of the parties, for
example for $N=3$ $P(H_{X^1_1 X^2_0 X^3_0})=H_{X^1_1 X^2_0 X^3_0}+H_{X^1_0 X^2_1 X^3_0}+H_{X^1_0 X^2_0 X^3_1}$ and  $P(H_{X^1_0 X^2_0})=H_{X^1_0 X^2_0}+H_{X^1_0 X^3_0}+H_{X^2_0 X^3_0}$. It is easy to see that \eqnref{Mn} is
violated by a generalization of the distribution \eqnref{p1} for more parties, given by
\begin{widetext}
\begin{equation}
\label{p1N}
p^N_1\left(  x^1,\dots,x^n | X^1,\dots,X^N \right)  =\left\{
\begin{array}{ll}
1/4 & \text{, } x^1\oplus \dots \oplus x^N =X^1 \dots X^N\\
0 & \text{, otherwise}%
\end{array}
\right. ,
\end{equation}
\end{widetext}
if we just mix it with the classical correlated distribution ($ x^1\oplus \cdots \oplus x^N =0$).

A nice feature of the entropic inequalities is that they can be readily applied to marginal scenarios with an arbitrary number of outcomes. This is in sharp contrast to the usual Bell inequalities approach where increasing
the number of outcomes also increases the complexity and dimension of the correlation polytopes. To our knowledge very few inequalities have been derived for marginal multipartite Bell scenarios with many outcomes, in
particular in Ref. \cite{Grandjean2012} tripartite inequalities have been derived for any number of outcomes, but as the authors stress there is no straightforward generalization of their methods to more parties (Also note
the Ref. \cite{Arnault2011}, but there the inequalities involve products of observables from the same party and therefore have no direct application to Bell scenarios). Entropic inequalities may be proven as a useful tool in
such cases. We have briefly explored this possibility by looking for quantum violations of the inequality \eqnref{Mn} using multidimensional GHZ states given by
\begin{equation}
\label{GHZd}
\ket{GHZ^N_d}= \sqrt{\frac{1}{d}} \sum_{i=0,\dots,d-1}\ket{i_1 \dots i_N}
\end{equation}
and employing the Fourier-transformed measurements used in \cite{CGLMP2002}. We have considered $N=3,\dots,10$ and $d=2,\dots,10$ and found that the violation of \eqnref{Mn} increases with both $N$ and $d$. As numerically
noted in \cite{Grandjean2012}, the maximal quantum violation for the inequalities considered there can be reached only using systems with local Hilbert space dimension exceeding the number of measurement outcomes; what
suggests that this kind of inequalities can be used as multipartite dimension witnesses \cite{Brunner2008}. We believe this is an interesting line of research one may pursue in the entropic approach.

\section{Computational Results}
\label{sec:computational_results}
In Sec. \ref{sec:Imm} we have used a specific combination of the basic inequalities in order to derive the entropic inequalities \eqref{I_I} and \eqref{BI_H}. However in principle different combinations could give rise to
different classes of entropic inequalities. To understand what other classes of inequalities one gets, we rely in this section on computational results. Using standard software to perform the Fourier-Motzkin elimination we
computed all classes of entropic inequalities for the simplest marginal models where the computation is expected to finish in a reasonable time.

It turns out that even for very simple scenarios involving more than $5$ variables the computations are already too large to finish. To go beyond that limitation one needs to further simplify the set of basic inequalities.
In order to do that we follow the approach proposed in Ref. \cite{Budroni2012} for usual Bell inequalities. Let us begin considering a bipartite scenario, with $m$ measurement settings for Alice and $n$ for Bob. The
existence of a classical description for all pairwise observables is equivalent to the existence of classical descriptions for the $n$ subsystems, $\left\{A_0,\cdots,A_{m-1},B_j\right\}$ with $j=0,\cdots,n-1$, coinciding on
$\left\{A_0,\cdots,A_{m-1}\right\}$. To find all the entropic inequalities for the marginal scenario it is then sufficient to start out with the union of the basic set of inequalities defining each one of the $n$ subsystems
(see Fig. \ref{fig:pieces1}). A further simplification is possible since for each of the subsystems (indexed by $j$) it is sufficient to consider only the inequalities involving the subsets of the variables $\left\{
\left\{A_i,B_j\right\} ,\left\{A_0,\cdots,A_m\right\} \right\}$ with $i=0,\dots, m-1$. That is, in practice we start with the set of Shannon-type inequalities describing the cone $\left\{A_0,\cdots,A_{m-1},B_j\right\}$ and
project it down to the to the cone describing $\left\{ \left\{A_i,B_j\right\}, \left\{A_0,\cdots,A_m\right\} \right\}$. With this simplification we were able to fully characterize bipartite marginal models involving up to
$7$ variables, also accounting for the effects of bounded shared randomness.

Using similar simplifications, we also obtain inequalities for marginal scenarios involving statistical independencies. Details are given in the Sec. \ref{subsec:mar_stat_inde} below.

\begin{figure}[t!]
\includegraphics[width=0.8\linewidth]{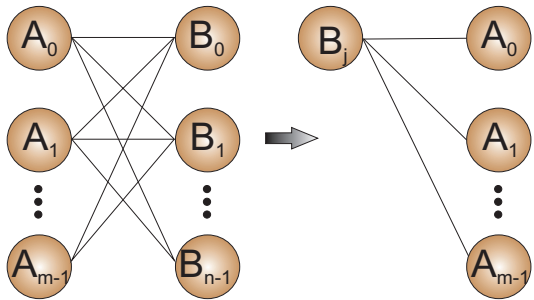}
\caption{(Color online) Graphical representation (figure on the left) of bipartite Bell scenarios: each vertex represents an observable and edges connected observables that are jointly measurable. The existence of a
classical description for all pairwise observables is equivalent to the existence of classical descriptions for the $n$ subsystems (figure on the right), $\left\{A_0,\cdots,A_{m-1},B_j\right\}$ with $j=0,\cdots,n-1$,
coinciding on $\left\{A_0,\cdots,A_{m-1}\right\}$.}\label{fig:pieces1}
\end{figure}

To characterize the entropic cone of a multipartite marginal model we can in principle proceed as before, first simplifying the set of basic inequalities. For example for $3$ parties, similarly to the bipartite case we can
restrict the initial set of inequalities to the ones describing the following two subsystems, $\left\{A_0,A_1,B_0,B_1,C_0\right\}$ and $\left\{A_0,A_1,B_0,B_1,C_1\right\}$, where $A_i$, $B_j$ and $C_k$ with $i,j,k= \left\{
0,1 \right\}$ describe the measurement choices available to the parties. However, even with that the FM elimination still demanded too many computational resources and we were not able to finish the computation. This
highlights the value of analytical derivations as the one in Sec. \ref{sec:Multi}.

\subsection{Bipartite scenario}
In the simplest case, given by the CHSH scenario ($m_a=m_b=2$, that is, $2$ measurement settings for Alice and Bob) it is known \cite{Chaves2012,FritzChaves2013} that the only class of non-trivial inequalities is given by
the entropic CHSH. Using the computational approach described above it follows that in the case with $(m_a=2,m_b=3)$ there are still only entropic CHSH inequalities, in full analogy with the probabilistic case. However,
differences to the probabilistic case already start to appear in the case $(m_a=m_b=3)$. For probabilities there are only two different classes of non-trivial inequalities, the CHSH and the $I_{3322}$ inequality
\cite{Collins2004}. However for entropies there are $4$ classes of non-trivial tight inequalities, shown in the table \ref{tab_bi}. Inequalities $3$ and $4$ correspond, respectively, to the $I^{E}_{22}$ and to the
$I^{E}_{33}$. By theorem 6 all the inequalities are also valid in the probability space after the proper translation is made, however, inequalities $5$ and $6$ do not correspond to tight inequalities in the correlation
polytope. For the scenario with $(m_a=3,m_b=4)$, $5$ new classes of inequalities have been found, the inequalities $7$ to $11$ in Table \ref{tab_bi}.

\begin{figure}
\includegraphics[width=0.8\linewidth]{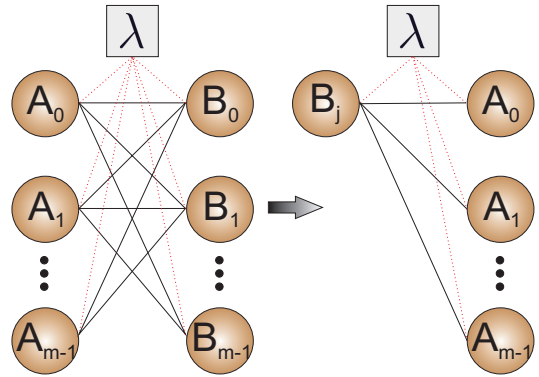}
\caption{(Color online) Graphical representation (figure on the left) of bipartite Bell scenarios: each vertex represents an observable and edges connect observables that are jointly measurable. The dotted red edges
represent the fact that all the correlations between the parties must be mediated by the hidden variable $\lambda$, that is, $I_{A:B| \lambda}=0$. The existence of a classical description for all pairwise observables is
equivalent to the existence of classical descriptions for the $n$ subsystems (figure on the right), $\left\{A_0,\cdots,A_{m-1},B_j,\lambda\right\}$ with $j=0,\cdots,n-1$, coinciding on
$\left\{A_0,\cdots,A_{m-1},\lambda\right\}$.}\label{fig:pieces2}
\end{figure}

We have also performed the same computation, but now bounding the amount of shared randomness, following the idea described in Sec. \ref{subsec:boundedlambda}. First of all we note that as one should expect, the inequalities
derived in the absence of any restriction on $H_{\lambda}$ still define facets of the entropic cone. For the inequalities depending on $H_{\lambda}$ it follows that in the $(m_a=m_b=2)$ and $(m_a=2,m_b=3)$ cases, the only
inequalities bounding the shared randomness are the ones given in Table \ref{tab_bi_BSR}, inequalities $1$ to $3$. For the case $(m_a=m_b=3)$ it follows that the problem is already too demanding and we were not able to
finish the computation. However, a further simplification is possible as shown in Fig. \ref{fig:pieces2}. We consider the union of the basic inequalities for the sets $\left\{A_0,A_1,A_{2},B_j, \lambda \right\}$
$\left\{A_j,B_0,B_1,B_2, \lambda \right\}$ with $j=0,1,2$, coinciding on $\left\{A_0,A_1,A_2\right\}$ and $\left\{B_0,B_1,B_2\right\}$. For each of these sets of inequalities parameterized by $j$, we first project it
down to the cones describing $\left\{ \left\{A_i,B_j\right\}, \left\{A_0,A_1,A_2,\lambda\right\} \right\}$ and $\left\{ \left\{A_i,B_j\right\}, \left\{B_0,B_1,B_2,\lambda\right\} \right\}$ with $i,j=0,1,2$. Then we
proceed to the final FM elimination, eliminating the non-observable variables. With this simplification we were able to finish the computation and $26$ new classes of inequalities were found, inequalities $4$ to $29$ in
Table \ref{tab_bi_BSR}.

All the inequalities in Table \ref{tab_bi_BSR} have a rather remarkable feature. For all of them it is not difficult to prove that the maximal value achievable by local correlations is given by $H_{B_0}+\min{ \left\{
H_{A_0},H_{B_0} \right\}}$. That is, up to $3$ measurements settings and considering dichotomic observables, any local distribution needs not more than $2$ bits of shared randomness to be simulated. As we discuss in Sec.
\ref{sec:discussion} it seems improbable that, increasing the number of measurement settings, only $2$ bits of shared randomness still would suffice to simulate any local distribution, specially the ones arising from
entangled states. It would be very interesting to find classes of entropic inequalities that would require in principle more bits of shared randomness to simulate local distributions.

\subsection{Scenarios with statistical independencies between the hidden variables}
\label{subsec:mar_stat_inde}

Consider three random variables $A$, $B$ and $C$, characterizing for instance some traits of three different languages \cite{Steudel2010}. From the observed data one concludes that all the three variables are all pairwise
maximally correlated, for example, $I_{A:B}=I_{A:C}=I_{B:C}=H_{A}=H_{B}=H_{C}=H_{ABC}$. Furthermore, no conditional independencies can be inferred from the data. The question is then: are the observed correlations compatible
with a causal structure involving no common ancestor to all the three variables (Fig. \ref{fig:triangle} on the right)? Or is a common ancestor needed to explain the data (Fig. \ref{fig:triangle} on the left)? A
direct application of causal discovery algorithms \cite{Pearlbook,Spirtesbook,Spekkens2012} would try to distinguish between the two causal structures, but since no conditional independencies are imposed and by the principle
of minimality (Occam's Razor), the algorithm would return the causal structure on the right as the answer. But clearly this is wrong, because this causal structure implies that (for the observed data) if $A$ is maximally
correlated with $B$, then it should be completely uncorrelated with $C$.

\begin{figure}[t!]
\includegraphics[width=0.9\linewidth]{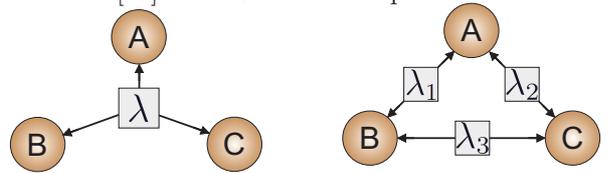}
\caption{If the correlations between three random variables are not sufficiently strong, they can be achieved by a causal structure involving no common ancestor to all the variables (left). Sufficiently strong correlations
can only be explained by the existence of common ancestor (left).}\label{fig:triangle}
\end{figure}

The entropic approach offers a surprisingly simple solution in this case. It is not difficult to show that the causal structure on the left of Fig. \ref{fig:triangle} implies a bound on the correlations given by
$I_{A:B}+I_{A:C} \leq H_{A}$ (and permutations thereof) \cite{Fritz2012}. The observed data clearly violates this inequality, meaning that it cannot be explained by the corresponding causal structure. That is, an ancestor
common to all the three variables is required to explain the observed probability distribution. It is interesting to note that this same scenario has been considered under two very different perspectives, from the purely
causal inference one \cite{Steudel2010} but also from the quantum non-locality point of view \cite{Fritz2012,Branciard2012}.

The causal structure depicted on the right of Fig. \ref{fig:triangle} imply many statistical independencies. As an example we have that $I_{\lambda_1:\lambda_2}=I_{\lambda_1:\lambda_3}=I_{\lambda_2:\lambda_3}=0$ and
$I_{A:B|\lambda_1}=I_{A:C|\lambda_2}=I_{B:C|\lambda_3}=0$.
Using all
available constraints many of the variables can be eliminated. A
final FM elimination gives as a result that $3$ different classes of
non-trivial entropic inequalities completely characterize the marginal
scenario:
\begin{widetext}
\begin{align}
\label{triangle_1}
& H_A+ H_B + H_C- H_{AB}-H_{AC}     \leq 0     \\
\label{triangle_2}
& 3H_A + 3H_B +3H_C -3H_{AB} -2H_{AC} -2H_{BC}+ H_{ABC} \leq 0  \\
\label{triangle_3}
& 5H_A + 5H_B +5H_C -4H_{AB} -4H_{AC} -4H_{BC}+ 2H_{ABC} \leq 0
\end{align}
\end{widetext}
Note that the inequality \eqref{triangle_1} is exactly the same as the one obtained in \cite{Fritz2012}. Our derivation shows that indeed this is a tight Shannon-type inequality. However, the are two other inequivalent classes, inequalities \eqref{triangle_2} and \eqref{triangle_3} that were not known before.

Another interesting case is the one of a common ancestor to all the variables, as depicted on the left of Fig. \ref{fig:triangle}, however now bounding the entropy of the common ancestor to be $H_{\lambda} \leq \mathcal{C}$.
In this case we find that the only non-trivial inequality is given by
\begin{equation}
\label{triangle_ancestor}
H_{AB}+H_{AC}+H_{BC}- 2H_{ABC} \leq  \mathcal{C}
\end{equation}
For a distribution fulfilling $I_{A:B}=I_{A:C}=I_{B:C}=H_{A}=H_{B}=H_{C}=H_{ABC}$ the maximal value of the expression \eqref{triangle_ancestor} is $H_{A}$, that is, as expected the distribution needs not more than $H_A$ bits of shared randomness to be achieved.

\section{Discussion}
\label{sec:discussion}

In this work we have explored the entropic approach to marginal problems, gathering several results that we believe may pave the way to a better understanding and more systematic application of entropic inequalities in a
wide range of applications. In the next paragraphs we summarize and briefly discuss our findings, with special attention to the open problems and possibilities that we believe deserve future investigation.

We have shown a correspondence between Shannon-type inequalities and inequalities in the probability space, stating that any Shannon-type inequality is also a valid probabilistic inequality if a very simple translation is
made. This correspondence formally explains the similarities observed for the n-cycle marginal scenario (that has as a particular case the CHSH scenario) \cite{Araujo2012,Chaves2012,FritzChaves2013} between the entropic
inequalities and the probabilistic version. For the n-cycle scenario all the non-trivial Shannon entropic inequalities have an exact correspondence in the probability space \cite{FritzChaves2013,Chaves2013}, however this is not
true in general, since not all probabilistic inequalities define valid Shannon entropic inequalities, that is, they involve probability inequalities that cannot be translated to a Shannon type entropic inequality. Also, as
mentioned before not all valid entropic inequalities are of the Shannon type. Could it be that taking into account non-Shannon type inequalities a deeper correspondence between entropic and probabilistic inequalities can be
made? The use of non-Shannon type inequalities is also interesting from a practical perspective, since in principle taking them into account one may get more restrictive inequalities, bounding more tightly the set of
allowed correlations.

Based on the correspondence between entropic and probabilistic inequalities we analytically proved the entropic version of the Collins-Gisin inequalities \cite{Collins2004}, valid for a bipartite scenario where each party
has access to $m$ measurements each ($m_a=m_b=m$). The entropic inequalities have the advantage of being valid for observables with any number of outcomes while the CG inequalities are specially tailored for dichotomic ones.
Moreover, for the scenarios with ($m_a=m_b=2$), ($m_a=m_b=3$) and $(m_a=3, m_b=4)$  we have computationally (through the FM elimination) derived all the entropic inequalities and shown that there are, respectively $3$, $6$,
$11$ inequivalent classes of inequalities. For the computational results, since the FM elimination generally produces as an output a huge list of redundant inequalities, we have also proven a result that allows one to check,
given the list of extremal points and half-lines of the convex set, if the inequality corresponds to a facet or not (see Appendix \ref{find_facets}).

We have also considered, for the bipartite case with ($m_a=m_b=2$) and ($m_a=m_b=3$), the effects of bounded shared randomness and shown that in these cases any local distribution with dichotomic outcomes can be entropically
simulated with at most two bits of shared randomness. This is a very interesting point that deserves further investigation. It seems implausible that increasing the number of measurement settings, any local distribution
would still require at most two bits of shared randomness. To further understand that, we have analytically proven a different entropic version of the CG inequalities (for arbitrary $m$), where the effects of bounded shared
randomness are taken into account. However, these inequalities still have the surprising property that no more than two bits of shared randomness are necessary to entropically simulate any local distribution. However, this
is only one class of inequalities and there are possibly many more with increasing $m$. If one can find other classes of inequalities one could investigate, for example, what are the shared randomness requirements to
simulate the correlations of a Werner state $\varrho_{\text{W}}=v\ket{\Phi^{+}}\bra{\Phi^{+}}+(1-v)\mathbb{I}/4$ parameterized by the visibility $v$ \cite{Werner1989}. In the region where the state is known to violate
some Bell inequality \cite{Vertesi2008}, since the state is nonlocal it follows that even an infinite amount of shared randomness is not sufficient to reproduce the correlations. There are however two interesting
regions where not much is known. First, for $ 1/3< v < 1/K_G \approx 0.66$ the state is entangled but local \cite{Acin2006}. Since the state is entangled one may expect that more shared randomness would be required,
but how much of it? The most interesting region is of course the one for $ 1/K_G< v < v_{\text{Ver}} \approx 0.7056 $ where it is not known if the state is nonlocal or not.
What are the shared randomness requirements for correlations obtained in this region?

Working in a generalization of the approach followed in \cite{Braunstein1988}, we derived entropic inequalities for multipartite marginal scenarios consisting of any number of parties, each having access to two observables
with any possible number of measurement outcomes. Using specific projective measurements we have numerically shown that the violation of these inequalities for multidimensional multipartite GHZ states increase with both the
size and the local dimension of the state. An interesting perspective would be the possible use of these inequalities as multipartite dimension witnesses, similarly to what has been suggested in \cite{Grandjean2012}.

Finally we have considered a scenario involving conditional independencies, for which the question is to decide if a given correlation for the observable quantities is compatible with a causal structure involving only
pairwise common ancestors. A natural question is how to generalize the obtained results to the case of many observable quantities and different configurations of common ancestors
\cite{Steudel2010}. An interesting related problem would be to understand relaxations over the bilocality assumption of entanglement swapping experiments \cite{Branciard2010}, for example, allowing correlations between the
hidden variables while keeping the bilocality on the level of the observed quantities. Similarly one could use entropic inequalities to relax the measurement independence assumption \cite{Hall2010,Barrett2010}, stating that
the measurement choice made by the parties is independent of the hidden variable.

\section{Acknowledgements}
It is a pleasure to thank Dominik Janzing for insightful discussions
about causal structures. We also would like to thank A.\ Ac\'in and
J.\ B.\ Brask for pointing out the potential application of bounded shared
randomness inequalities
for Werner states. Our work is supported by the Excellence Initiative of the German Federal and State Governments (Grant ZUK 43).

\bibliography{sharedrand}

\appendix
\label{sec:appendix}

\section{Tables with entropic inequalities}
\label{tables_ineqs}
\begin{table*}
\begin{tabular}{|c| c c c c| c c c c| c c c c c c c c c c c c|} \hline
\multicolumn{21}{|c|}{Entropic Bipartite Inequalities}\\
\hline
      \multicolumn{1}{|c|}{\textbf{\#}}
    & \multicolumn{4}{|c|}{\scriptsize$H(A_x)$}
    & \multicolumn{4}{|c|}{\scriptsize$H(B_y)$}
    & \multicolumn{12}{|c|}{\scriptsize$H(A_xB_y)$}\\
\tiny{$x$ / $y$ / $xy$}
&\tiny{0}&\tiny{1}&\tiny{2}&\tiny{3}
&\tiny{0}&\tiny{1}&\tiny{2}&\tiny{3}
&\tiny{00}&\tiny{01}&\tiny{02}&\tiny{03}&\tiny{10}&\tiny{11}&\tiny{12}&\tiny{13}&\tiny{20}&\tiny{21}&\tiny{22}&\tiny{23}
\\
\hline
\textbf{1}
&\tiny{1}&\tiny{0}&\tiny{0}&\tiny{0}
&\tiny{0}&\tiny{0}&\tiny{0}&\tiny{0}
&\tiny{-1}&\tiny{0}&\tiny{0}&\tiny{0}&\tiny{0}&\tiny{0}&\tiny{0}&\tiny{0}&\tiny{0}&\tiny{0}&\tiny{0}&\tiny{0}
\\
\hline
\textbf{2}
&\tiny{-1}&\tiny{0}&\tiny{0}&\tiny{0}
&\tiny{-1}&\tiny{0}&\tiny{0}&\tiny{0}
&\tiny{1}&\tiny{0}&\tiny{0}&\tiny{0}&\tiny{0}&\tiny{0}&\tiny{0}&\tiny{0}&\tiny{0}&\tiny{0}&\tiny{0}&\tiny{0}
\\
\hline
\textbf{3}
&\tiny{1}&\tiny{0}&\tiny{0}&\tiny{0}
&\tiny{1}&\tiny{0}&\tiny{0}&\tiny{0}
&\tiny{1}&\tiny{1}&\tiny{0}&\tiny{1}&\tiny{-1}&\tiny{0}&\tiny{0}&\tiny{0}&\tiny{0}&\tiny{0}&\tiny{0}&\tiny{0}
\\
\hline
\textbf{4}
&\tiny{1}&\tiny{0}&\tiny{0}&\tiny{0}
&\tiny{2}&\tiny{1}&\tiny{0}&\tiny{0}
&\tiny{-1}&\tiny{-1}&\tiny{-1}&\tiny{-1}&\tiny{-1}&\tiny{1}&\tiny{-1}&\tiny{1}&\tiny{0}&\tiny{0}&\tiny{0}&\tiny{0}
\\
\hline
\textbf{5}
&\tiny{1}&\tiny{1}&\tiny{0}&\tiny{0}
&\tiny{1}&\tiny{1}&\tiny{0}&\tiny{0}
&\tiny{-1}&\tiny{1}&\tiny{-1}&\tiny{0}&\tiny{-1}&\tiny{-1}&\tiny{-1}&\tiny{-1}&\tiny{+1}&\tiny{0}&\tiny{0}&\tiny{0}
\\
\hline
\textbf{6}
&\tiny{1}&\tiny{1}&\tiny{0}&\tiny{0}
&\tiny{2}&\tiny{0}&\tiny{0}&\tiny{0}
&\tiny{-1}&\tiny{-1}&\tiny{-1}&\tiny{-1}&\tiny{-1}&\tiny{1}&\tiny{-2}&\tiny{1}&\tiny{0}&\tiny{0}&\tiny{0}&\tiny{0}
\\
\hline
\textbf{7}
&\tiny{2}&\tiny{1}&\tiny{0}&\tiny{0}
&\tiny{1}&\tiny{1}&\tiny{0}&\tiny{0}
&\tiny{-1}&\tiny{0}&\tiny{-1}&\tiny{-1}&\tiny{0}&\tiny{-1}&\tiny{-1}&\tiny{1}&\tiny{-1}&\tiny{-1}&\tiny{1}&\tiny{0}
\\
\hline
\textbf{8}
&\tiny{1}&\tiny{1}&\tiny{0}&\tiny{0}
&\tiny{1}&\tiny{1}&\tiny{1}&\tiny{0}
&\tiny{-1}&\tiny{-1}&\tiny{1}&\tiny{0}&\tiny{-1}&\tiny{0}&\tiny{-1}&\tiny{-1}&\tiny{0}&\tiny{-1}&\tiny{-1}&\tiny{1}
\\
\hline
\textbf{9}
&\tiny{1}&\tiny{1}&\tiny{0}&\tiny{0}
&\tiny{2}&\tiny{1}&\tiny{1}&\tiny{0}
&\tiny{-1}&\tiny{-1}&\tiny{1}&\tiny{-1}&\tiny{-1}&\tiny{1}&\tiny{-1}&\tiny{-1}&\tiny{-1}&\tiny{-1}&\tiny{-1}&\tiny{1}
\\
\hline
\textbf{10}
&\tiny{2}&\tiny{1}&\tiny{0}&\tiny{0}
&\tiny{2}&\tiny{1}&\tiny{0}&\tiny{0}
&\tiny{-1}&\tiny{-1}&\tiny{-1}&\tiny{-1}&\tiny{-1}&\tiny{-1}&\tiny{-1}&\tiny{1}&\tiny{-2}&\tiny{1}&\tiny{1}&\tiny{0}
\\
\hline
\textbf{11}
&\tiny{2}&\tiny{1}&\tiny{0}&\tiny{0}
&\tiny{1}&\tiny{1}&\tiny{1}&\tiny{0}
&\tiny{-1}&\tiny{-2}&\tiny{0}&\tiny{-1}&\tiny{-1}&\tiny{1}&\tiny{-1}&\tiny{-1}&\tiny{1}&\tiny{-1}&\tiny{-1}&\tiny{1}
\\
\hline
\end{tabular}
\caption{All classes of entropic bipartite inequalities for $m_a=2,3$ and $m_b=2,3,4$. We have
listed the coefficients of one inequality in each row, and all
inequalities are of the form $\leq 0$.} \label{tab_bi}
\end{table*}

\begin{table*}
\begin{tabular}{|c| c c c| c c c| c c c c c c c c c| c|} \hline
\multicolumn{17}{|c|}{Entropic Bipartite Inequalities with bounded Shared Randomness}\\
\hline
      \multicolumn{1}{|c|}{\textbf{\#}}
    & \multicolumn{3}{|c|}{\scriptsize$H(A_x)$}
    & \multicolumn{3}{|c|}{\scriptsize$H(B_y)$}
    & \multicolumn{9}{|c|}{\scriptsize$H(A_xB_y)$}
    & \multicolumn{1}{|c|}{\textbf{Bound }}\\
\tiny{$x$ / $y$ / $xy$}
&\tiny{0}&\tiny{1}&\tiny{2}
&\tiny{0}&\tiny{1}&\tiny{2}
&\tiny{00}&\tiny{01}&\tiny{02}&\tiny{10}&\tiny{11}&\tiny{12}&\tiny{20}&\tiny{21}&\tiny{22}
&$c$
\\
\hline
\textbf{1}
&\tiny{1}&\tiny{0}&\tiny{0}
&\tiny{1}&\tiny{0}&\tiny{0}
&\tiny{-1}&\tiny{0}&\tiny{0}&\tiny{0}&\tiny{0}&\tiny{0}&\tiny{0}&\tiny{0}&\tiny{0}
&$1$
\\
\hline
\textbf{2}
&\tiny{1}&\tiny{0}&\tiny{0}
&\tiny{1}&\tiny{0}&\tiny{0}
&\tiny{0}&\tiny{-1}&\tiny{-1}&\tiny{1}&\tiny{0}&\tiny{0}&\tiny{0}&\tiny{0}&\tiny{0}
&$1$
\\
\hline
\textbf{3}
&\tiny{1}&\tiny{1}&\tiny{0}
&\tiny{1}&\tiny{0}&\tiny{0}
&\tiny{-1}&\tiny{-1}&\tiny{1}&\tiny{-1}&\tiny{1}&\tiny{-1}&\tiny{0}&\tiny{0}&\tiny{0}
&$1$
\\
\hline
\textbf{4}
&\tiny{1}&\tiny{1}&\tiny{0}
&\tiny{1}&\tiny{0}&\tiny{0}
&\tiny{-1}&\tiny{-1}&\tiny{1}&\tiny{0}&\tiny{0}&\tiny{-1}&\tiny{-1}&\tiny{1}&\tiny{0}
&$1$
\\
\hline
\textbf{5}
&\tiny{1}&\tiny{1}&\tiny{0}
&\tiny{1}&\tiny{1}&\tiny{0}
&\tiny{-1}&\tiny{1}&\tiny{-1}&\tiny{0}&\tiny{-1}&\tiny{-1}&\tiny{0}&\tiny{-1}&\tiny{1}
&$1$
\\
\hline
\textbf{6}
&\tiny{2}&\tiny{0}&\tiny{0}
&\tiny{1}&\tiny{1}&\tiny{0}
&\tiny{-1}&\tiny{0}&\tiny{-2}&\tiny{-1}&\tiny{0}&\tiny{1}&\tiny{1}&\tiny{-1}&\tiny{0}
&$1$
\\
\hline
\textbf{7}
&\tiny{2}&\tiny{1}&\tiny{0}
&\tiny{1}&\tiny{1}&\tiny{0}
&\tiny{-1}&\tiny{1}&\tiny{-2}&\tiny{0}&\tiny{-1}&\tiny{-1}&\tiny{-1}&\tiny{-1}&\tiny{2}
&$1$
\\
\hline
\textbf{8}
&\tiny{2}&\tiny{1}&\tiny{0}
&\tiny{1}&\tiny{1}&\tiny{0}
&\tiny{-1}&\tiny{0}&\tiny{-2}&\tiny{1}&\tiny{-1}&\tiny{-1}&\tiny{-1}&\tiny{-1}&\tiny{2}
&$1$
\\
\hline
\textbf{9}
&\tiny{2}&\tiny{1}&\tiny{0}
&\tiny{1}&\tiny{1}&\tiny{0}
&\tiny{-2}&\tiny{1}&\tiny{-1}&\tiny{1}&\tiny{-1}&\tiny{-1}&\tiny{-1}&\tiny{-1}&\tiny{1}
&$1$
\\
\hline
\textbf{10}
&\tiny{1}&\tiny{1}&\tiny{1}
&\tiny{1}&\tiny{1}&\tiny{0}
&\tiny{-1}&\tiny{-2}&\tiny{1}&\tiny{-1}&\tiny{1}&\tiny{-1}&\tiny{1}&\tiny{-1}&\tiny{-1}
&$1$
\\
\hline
\textbf{11}
&\tiny{2}&\tiny{1}&\tiny{0}
&\tiny{1}&\tiny{1}&\tiny{0}
&\tiny{-1}&\tiny{1}&\tiny{-2}&\tiny{-1}&\tiny{-1}&\tiny{1}&\tiny{1}&\tiny{-1}&\tiny{0}
&$1$
\\
\hline
\textbf{12}
&\tiny{2}&\tiny{1}&\tiny{0}
&\tiny{1}&\tiny{1}&\tiny{0}
&\tiny{-2}&\tiny{1}&\tiny{-1}&\tiny{-2}&\tiny{-1}&\tiny{1}&\tiny{1}&\tiny{-1}&\tiny{0}
&$1$
\\
\hline
\textbf{13}
&\tiny{2}&\tiny{2}&\tiny{0}
&\tiny{2}&\tiny{1}&\tiny{0}
&\tiny{-2}&\tiny{1}&\tiny{-2}&\tiny{1}&\tiny{-1}&\tiny{-2}&\tiny{-2}&\tiny{-1}&\tiny{2}
&$1$
\\
\hline
\textbf{14}
&\tiny{2}&\tiny{1}&\tiny{0}
&\tiny{1}&\tiny{1}&\tiny{0}
&\tiny{-3}&\tiny{-1}&\tiny{1}&\tiny{-1}&\tiny{1}&\tiny{-1}&\tiny{1}&\tiny{-1}&\tiny{0}
&$1$
\\
\hline
\textbf{15}
&\tiny{3}&\tiny{1}&\tiny{0}
&\tiny{2}&\tiny{1}&\tiny{0}
&\tiny{-2}&\tiny{1}&\tiny{-3}&\tiny{1}&\tiny{-1}&\tiny{-1}&\tiny{-2}&\tiny{-1}&\tiny{2}
&$1$
\\
\hline
\textbf{16}
&\tiny{2}&\tiny{1}&\tiny{0}
&\tiny{1}&\tiny{1}&\tiny{0}
&\tiny{-1}&\tiny{1}&\tiny{-2}&\tiny{-1}&\tiny{-1}&\tiny{1}&\tiny{1}&\tiny{-1}&\tiny{0}
&$2$
\\
\hline
\textbf{17}
&\tiny{2}&\tiny{1}&\tiny{0}
&\tiny{1}&\tiny{1}&\tiny{0}
&\tiny{-2}&\tiny{1}&\tiny{-1}&\tiny{1}&\tiny{-1}&\tiny{-1}&\tiny{0}&\tiny{-1}&\tiny{1}
&$2$
\\
\hline
\textbf{18}
&\tiny{2}&\tiny{1}&\tiny{0}
&\tiny{1}&\tiny{1}&\tiny{0}
&\tiny{1}&\tiny{-1}&\tiny{-2}&\tiny{-2}&\tiny{1}&\tiny{0}&\tiny{0}&\tiny{-1}&\tiny{1}
&$2$
\\
\hline
\textbf{19}
&\tiny{2}&\tiny{2}&\tiny{0}
&\tiny{1}&\tiny{1}&\tiny{0}
&\tiny{-2}&\tiny{1}&\tiny{-1}&\tiny{1}&\tiny{-1}&\tiny{-2}&\tiny{-1}&\tiny{-1}&\tiny{2}
&$2$
\\
\hline
\textbf{20}
&\tiny{3}&\tiny{1}&\tiny{0}
&\tiny{1}&\tiny{1}&\tiny{0}
&\tiny{-2}&\tiny{1}&\tiny{-2}&\tiny{1}&\tiny{-1}&\tiny{-1}&\tiny{-1}&\tiny{-1}&\tiny{2}
&$2$
\\
\hline
\textbf{21}
&\tiny{2}&\tiny{1}&\tiny{1}
&\tiny{1}&\tiny{1}&\tiny{0}
&\tiny{-2}&\tiny{1}&\tiny{-1}&\tiny{-2}&\tiny{-1}&\tiny{2}&\tiny{1}&\tiny{-1}&\tiny{-1}
&$2$
\\
\hline
\textbf{22}
&\tiny{3}&\tiny{1}&\tiny{0}
&\tiny{1}&\tiny{1}&\tiny{0}
&\tiny{-2}&\tiny{1}&\tiny{-2}&\tiny{-2}&\tiny{1}&\tiny{2}&\tiny{1}&\tiny{-1}&\tiny{0}
&$2$
\\
\hline
\textbf{23}
&\tiny{3}&\tiny{2}&\tiny{0}
&\tiny{2}&\tiny{1}&\tiny{0}
&\tiny{1}&\tiny{-2}&\tiny{-2}&\tiny{-2}&\tiny{2}&\tiny{-2}&\tiny{-2}&\tiny{-1}&\tiny{2}
&$2$
\\
\hline
\textbf{24}
&\tiny{3}&\tiny{1}&\tiny{0}
&\tiny{1}&\tiny{1}&\tiny{0}
&\tiny{-1}&\tiny{1}&\tiny{-3}&\tiny{1}&\tiny{-1}&\tiny{-1}&\tiny{-1}&\tiny{-1}&\tiny{2}
&$2$
\\
\hline
\textbf{25}
&\tiny{2}&\tiny{1}&\tiny{1}
&\tiny{1}&\tiny{1}&\tiny{0}
&\tiny{-1}&\tiny{-3}&\tiny{2}&\tiny{-1}&\tiny{1}&\tiny{-1}&\tiny{1}&\tiny{-1}&\tiny{-1}
&$2$
\\
\hline
\textbf{26}
&\tiny{2}&\tiny{2}&\tiny{0}
&\tiny{1}&\tiny{1}&\tiny{0}
&\tiny{-3}&\tiny{-1}&\tiny{2}&\tiny{-1}&\tiny{1}&\tiny{-2}&\tiny{1}&\tiny{-1}&\tiny{0}
&$2$
\\
\hline
\textbf{27}
&\tiny{3}&\tiny{1}&\tiny{0}
&\tiny{1}&\tiny{1}&\tiny{0}
&\tiny{-1}&\tiny{1}&\tiny{-3}&\tiny{-2}&\tiny{-1}&\tiny{2}&\tiny{1}&\tiny{-1}&\tiny{0}
&$2$
\\
\hline
\textbf{28}
&\tiny{3}&\tiny{1}&\tiny{0}
&\tiny{2}&\tiny{1}&\tiny{0}
&\tiny{1}&\tiny{-1}&\tiny{-3}&\tiny{-1}&\tiny{-1}&\tiny{0}&\tiny{-2}&\tiny{1}&\tiny{1}
&$2$
\\
\hline
\textbf{29}
&\tiny{3}&\tiny{2}&\tiny{0}
&\tiny{2}&\tiny{1}&\tiny{0}
&\tiny{-2}&\tiny{2}&\tiny{-3}&\tiny{1}&\tiny{-2}&\tiny{-1}&\tiny{-2}&\tiny{-1}&\tiny{2}
&$2$
\\
\hline
\end{tabular}
\caption{All classes of entropic bipartite inequalities with bounded shared randomness for $m_a=2,3$ and $m_b=2,3$. We have listed the coefficients of one inequality in each row, and all
inequalities are of the form $\leq c$.}\label{tab_bi_BSR}
\end{table*}

\section{Deciding wether an entropic inequality corresponds to a facet of the entropic cone}
\label{find_facets}

Every convex set can be expressed in terms of dual representations, either in terms of extremal points and half-lines or in terms of inequalities (half-spaces) defining the facets of the convex set. As discussed in Sec.
\ref{sec:shannon_cone} not much is known about the extremal half-lines/points of the Shannon-type entropic cone and in order to derive entropic inequalities for a given marginal scenario one needs to rely on the FM
elimination. One problem that arises is that after performing the FM elimination usually the set of inequalities will contain many (for the scenarios we consider in Sec. \ref{sec:computational_results} typically several
thousands) redundant inequalities not corresponding to facets of the cone. That is, among the huge list of inequalities obtained via the FM elimination, we need to find the minimal set of inequalities describing the marginal
scenario, that ones corresponding to facets of the entropic cone. One way to find this minimal set of inequalities is to solve a linear problem, i.e, whenever a given inequality can be expressed as a linear combination (with
positive coefficients) of other inequalities it can be safely eliminated. However, given the typical case we face, of sets containing a huge number of inequalities, this approach soon becomes unfeasible. Notwithstanding the
difficulty in characterizing the extremal rays/points of the Shannon-type entropic cone, for most of the marginal scenarios we consider computationally, we were also able to get a list of them. In order to derive the minimal
set of inequalities and further understand the structure of the entropic cones we rely instead on the information provided by the extremal points and half-lines.

Given the extreme points, extreme directions and a list of inequalities satisfied by all points of some polyhedral set it is easy to decide which of these inequalities belong to a minimal list characterizing the polyhedral
set.

Before we prove this fact in general let us first look at the example of the two-dimensional unbounded closed polyhedral set in two-dimensional Euclidean space that is defined by the inequalities
\begin{equation}
 x\geq0 \text{, }y\geq0 \text{ and }-x+y+1\geq0
\end{equation}
displayed in \ref{fig:convexset}.
Its extreme points are $p_{1}=(0,0)$ and $p_{2}=(1,0)$ and the extreme directions are $v_{1}=(0,1)$ and $v_{2}=(1,1)$. Given a list of the three defining inequalities plus lets say the valid inequalities $x+1\geq0$ and
$x+y\geq0$ we want to reduce it to the minimal list containing only the three defining inequalities.

\begin{figure}[t!]
\includegraphics[width=0.8\linewidth]{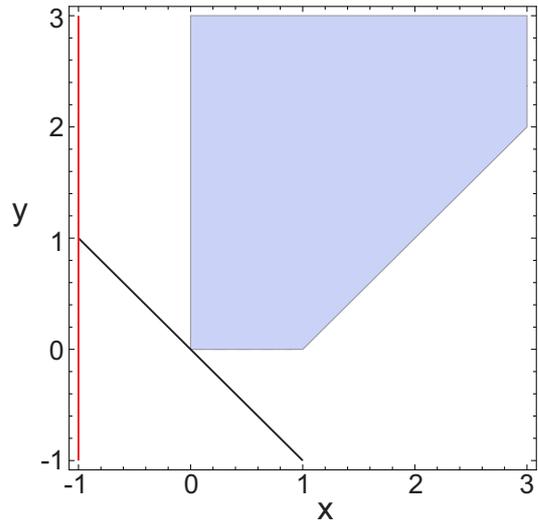}
\caption{(Color online) The convex set (blue color) generated by the half-lines $x\geq0 \text{, }y\geq0 \text{ and }-x+y+1\geq0$, where each of this inequalities generates a facet of the set. The inequalities $x+1\geq0$ and
$x+y\geq0$ are indicated, respectively, in red and black, and it is clear that while valid (since they can be obtained as a linear combination of the generating set) they are not facets of the proposed convex set.
}\label{fig:convexset}
\end{figure}

In this two-dimensional example, one can guess that a necessary and sufficient condition for a valid inequality belonging to the minimal list is that it is saturated by an one-dimensional subset of the polyhedral set.

We start with $x\geq 0$ and see that it is saturated by $p_{1}$ and the halfline $p_{1}+\lambda v_{1}$, $\lambda\geq 0$. The intersection of our polyhedral set with all points
in $\R^{2}$ saturating the inequality is therefore the halfline $p_{1}+\lambda v_{1}$. Its dimension is one and thus the inequality belongs to the minimal list.

Consider now the inequality $y\geq0$. We find $p_{1}$ and $p_{2}$ saturating it and thus the one-dimensional object $conv\{p_{1},p_{2}\}$ being the set of points saturating the inequality. The third defining inequality is
also saturated by an one-dimensional subset of our set - namely the halfline $p_{2}+\lambda v_{2}$, $\lambda\geq 0$.

In turn, the inequality $x+y\geq0$ is only saturated by $p_{1}$. So we found a zero-dimensional subset and disregard this inequality. The last inequality  $x+1\geq0$ is not saturated by any point, so we found the $\emptyset$
and also disregard this inequality. Clearly, the inequalities not belonging to the minimal list are already implied by those belonging to the minimal list.

To begin with the general case we note that every closed convex set $P$ in $\R^n$ that contains no lines is the convex hull of its extreme points and extreme half-lines (corollary 2.6.15 of \cite{Webster1995}). In our case,
however, after the FM elimination performed with PORTA \cite{porta}, it is not the extreme points and extreme half-lines we are given, but the extreme points and extreme directions. It is obvious that the set containing all
half-lines arising from the combination of all extreme points with all extreme directions contains the set of all extreme half-lines, as a half-line can not be extreme if it arises from the combination of a non extreme point
with a extreme direction. Formally speaking if $x=\alpha x_{1}+\beta x_{2}$ with $\alpha \textit{, } \beta \geq 0$ and $\alpha+\beta=1$ then $hl=x+\lambda y=\alpha (x_{1} + \lambda y ) + \beta (x_{2}+\lambda y)=\alpha
hl_{1}+\beta hl_{2}$, with extreme points $x_{1}$, $x_{2}$, extreme direction $y$ and extreme half-lines $hl_{1}$, $hl_{2}$.

Every inequality $\langle x|h_{i}\rangle\leq C_{i}$ corresponds to a hyper-plane $H_{i}=\{x|\langle x|h_{i}\rangle = C_{i}\}$ that divides the space into two closed half-spaces
\begin{equation}
 H_{i}^{\pm}=\{x|\langle x|h_{i}\rangle\leq \mp C_{i}\} \text{  with  }P\subseteq H_{i}^{-} \text{  for every $i$}
\end{equation}

\begin{definition}
Let $P$ be a closed convex set with dimension $r$ and $H$ some hyperplane with $P\subseteq H^{-}$ then $F:=P\cap H$ is called an exposed face. An exposed face with dimension $r-1$ is called a facet.
\end{definition}

Note that the name \textit{exposed face} is justified by the fact that $F$ indeed is a face of $P$.

With the following lemma it is easy to check wether an inequality corresponds to a facet.

\begin{lemma}
 Let $P=conv\{A\}$ with $A$ being the set of extreme points and extreme
 half-lines of $P$ and $H$ some hyperplane with $P\subseteq H^{-}$ then $F:=P\cap H=conv\{A\cap H\}$.
\end{lemma}

\begin{proof}
 \begin{eqnarray}
 & F=P\cap H=conv\{A\}\cap H= \\ \nonumber
 & \{x|x=\sum_{k} \alpha_{k} y_{k}\ , \langle\sum_{k} \alpha_{k} y_{k}|h\rangle=C, \\ \nonumber
 & \alpha_{k}\geq 0, \sum_{k} \alpha_{k}=1, y_{k}\in A \}
 \end{eqnarray}

 We know that $P\subseteq H^{-}$ and so $A\subseteq H^{-}$ which means that  $\langle y_{k}|h\rangle\leq C$ for every $k$. Thus the condition
 \begin{equation}
  \langle\sum_{k} \alpha_{k} y_{k}|h\rangle=C
 \end{equation}
can only be fulfilled for every convex combination if $\langle y_{k}|h\rangle=C$ for all $k$. And so

\begin{eqnarray}
& F=P\cap H=conv\{A\}\cap H=\{x|x= \\ \nonumber
& \sum_{k} \alpha_{k} y_{k}\ , \langle y_{k}|h\rangle=C,\alpha_{k}\geq 0, \\ \nonumber
& \sum_{k} \alpha_{k}=1, y_{k}\in A \}=conv\{A\cap H\}
 \end{eqnarray}
\end{proof}

For every inequality it is easy to find the set $A\cap H=\{x|\langle x|h\rangle=C, x\in A \}$ i.e. the set of all extreme points and half-lines of $P$ that saturate the inequality.

Still we have to check the dimension of $F$. It is known (theorem 4.1.3 of \cite{Webster1995}) that $F$ as the convex combination of its extreme points and extreme half-lines can be written as the direct sum of the convex
combination of its extreme points and the cone of its extreme directions :
\begin{eqnarray}
& F=\{x|x=\sum_{i=1}^{n}\lambda_{i}p_{i}+\sum_{i=1}^{m}\mu_{i}v_{i}\}= \\ \nonumber
& \{x|x=p_{1}+\sum_{i=2}^{n}\lambda_{i}(p_{i}-p_{1})+\sum_{i=1}^{m}\mu_{i}v_{i}\}
\end{eqnarray}
with $\sum_{i=1}^{n}\lambda_{i}=1$, $\lambda_{i}\geq 0$, $\mu_{i}\geq 0$, $p_{i}$ being the $n$ extreme points and $v_{i}$ the $m$ extreme directions.

We see that $F$ is some linear combination of the vectors ($p_{i}-p_{1}$) and $v_{i}$. The dimension of $F$ thus is equal to the rank of the matrix with columns ($p_{i}-p_{1}$) and $v_{i}$.

In summary what one has to do is to find all extreme points and directions that saturate a given inequality and calculate the rank of this matrix. If and only if it is equal to $(r-1)$ the inequality induces a facet.

\begin{widetext}

\section{Proving the $I^{E}_{mm} \leq 0$}
\label{Imm_app}

The proof will consist of two main steps. In the first step we will show how to obtain all the terms of the first two rows and the terms of the last row plus the marginal $H(A_0)$ in \eqref{Imm_ent}. We begin with the
SM(sub-modularity) inequality
\begin{equation}
H_{A_0A_1 \cdots A_{m-2}B_0}+H_{A_0A_1 \cdots,A_{m-2} B_1} \geq H_{A_0A_1 \cdots A_{m-2}B_0B_1}+_{A_0A_1\cdots A_{m-2}}.
\end{equation}
We then cancel the two NO (non-observable) terms on the LHS by adding the following SM inequalities
\begin{align}
& H_{A_0B_0}+H_{A_1A_2\cdots A_{m-2}B_0} \geq H_{A_0A_1\cdots A_{m-2}B_0}+H_{B_0} \\
& H_{A_0B_1}+H_{A_1A_2\cdots A_{m-2}B_1} \geq H_{A_0A_1\cdots A_{m-2}B_1}+H_{B_1}.
\end{align}
We continue to cancel the two NO terms on the LHS introduced by the last inequalities, by adding the following SM inequalities
\begin{align}
& H_{A_1 B_0}+H_{A_2 A_3 \cdots A_{m-2} B_0} \geq H_{A_1 A_2 \cdots A_{m-2} B_0}+H_{B_0} \\
& H_{A_1 B_1}+H_{A_2 A_3 \cdots A_{m-2} B_1} \geq H_{A_1 A_2 \cdots A_{m-2} B_1}+H_{B_1}.
\end{align}
Note that in each of these steps we get the O (observable) terms $\left\{ H_{A_i,B_0},H_{A_i,B_1},H_{B_0},H_{B_1} \right\}$. We continue doing that until we add the inequalities
\begin{align}
& H_{A_{m-4} B_0}+H_{A_{m-3} A_{m-2} B_0} \geq H_{A_{m-4} A_{m-3} A_{m-2} B_0}+H_{B_0} \\
& H_{A_{m-4} B_1}+H_{A_{m-3} A_{m-2} B_1} \geq H_{A_{m-4} A_{m-3} A_{m-2} B_1}+H_{B_1}.
\end{align}
Now we add the inequalities
\begin{align}
& H_{A_{m-3} B_0}+H_{A_{m-2} B_0} \geq H_{A_{m-3} A_{m-2} B_0}+H_{B_0} \\
& H_{A_{m-3} B_1}+H_{A_{m-2} B_1} \geq H_{A_{m-3} A_{m-2} B_1}+H_{B_1}.
\end{align}
In these $m-2$ steps we have added all the terms $\left\{ H_{A_i B_0},H_{A_i B_1}\right\}$ for $i=0,\cdots,m-2$ and $m-2$ $\left\{H_{B_0},H_{B_1} \right\}$. To complete the two first rows we add the SM inequality
\begin{equation}
H_{A_{m-1} B_0}+H_{B_0 B_1} \geq H_{A_{m-1} B_1}+H_{B_0}.
\end{equation}
To obtain the last row plus the marginal $H(A_0)$ we add the SM inequality
\begin{equation}
H_{A_0 B_{m-1}}+H_{A_0 A_1}\geq H_{A_1 B_{m-1}}+H_{A_0}.
\end{equation}
The remaining NO terms are
\begin{equation}
H_{A_0 A_1}+H_{B_0 B_1}\geq H_{A_0 A_1 \cdots A_{m-2} B_0 B_1}+H_{A_0 A_1 \cdots A_{m-2}}.
\end{equation}
We can cancel two of them by simply adding the M (monotonicity) inequality
\begin{equation}
H_{A_0 A_1 \cdots A_{m-2} B_0 B_1}\geq H_{B_0 B_1}.
\end{equation}
The remaining NO terms are then
\begin{equation}
H_{A_0 A_1} \geq H_{A_0 A_1 \cdots A_{m-2}}.
\end{equation}

In order to conclude the proof we need to show that the remaining $m-3$ rows minus the NO terms above define a valid inequality. To show that, first note that every missing row ($j=2,\cdots,m-2$) is of the form
\begin{equation}
\label{secondstep}
\sum_{i=0,\cdots,m-j} (-1)^{\delta_{i,m-j}} H_{A_i B_j} - (m-j-1)H_{B_j}.
\end{equation}
First we are going to prove that
\begin{equation}
\label{secondstep_inter}
\sum_{i=0,\cdots,m-j} (-1)^{\delta_{i,m-j}} H_{A_i B_j} - (m-j-1)H_{B_j}+ H_{A_0 \cdots A_{m-j}}- H_{A_0 \cdots A_{m-j-1}} \geq 0.
\end{equation}
To do that we add the following inequalities
\begin{align}
&H_{A_0 \cdots A_{m-j}}+H_{A_0 \cdots,A_{m-j-1} B_j} - H_{A_{m-j} B_j} - H_{A_0 \cdots A_{m-j-1}}\\
&H_{A_{m-j-1} B_j}+H_{A_0 \cdots A_{m-j-2} B_j} - H_{A_0 \cdots A_{m-j-1},B_j} -H_{B_j}\\
&H_{A_{m-j-2} B_j}+H_{A_0 \cdots,A_{m-j-3} B_j} - H_{A_0 \cdots A_{m-j-2} B_j} -H_{B_j}\\ \nonumber
& \cdots \\
&H_{A_2 B_j}+H_{A_0 A_1 B_j} - H_{A_0 A_1 A_2 B_j} -H_{B_j}\\
&H_{A_0 B_j}+H_{A_1 B_j} - H_{A_0 A_1 B_j} -H_{B_j}.
\end{align}
To finish the proof we only need to add the inequalities \eqref{secondstep_inter} for $j=2,\cdots,m-2$, that is
\begin{align}
& \sum_{j=2,\cdots,m-2}\sum_{i=0,\cdots,m-j} (-1)^{\delta_{i,m-j}} H_{A_i B_j} - (m-j-1)H_{B_j}+ H_{A_0 \cdots A_{m-j+1}}- H_{A_0 \cdots A_{m-j}} \geq 0 \\
& \rightarrow \sum_{j=2,\cdots,m-2} \sum_{i=0,\cdots,m-j} (-1)^{\delta_{i,m-j}} H_{A_i B_j} - (m-j-1)H_{B_j} - H_{A_0 A_1} + H_{A_0 A_1 \cdots A_{m-2}} \geq 0.
\end{align}

This concludes the proof.

\section{Proving the $BI^{E}_{mm} \leq \mathcal{C}$ inequality}
\label{Imm_bounded_app}
First remember that bounding the shared randomness between the parties implies the following constraint
\begin{align}
& H_{A_0 A_1 \cdots A_{m-1}}+H_{B_0 B_1 \cdots B_{m-1}}-H_{A_0 A_1 \cdots A_{m-1} B_0 B_1 \cdots B_{m-1}} \leq H_{\lambda}\\
& I_{A_0 A_1 \cdots A_{m-1}:B_0 B_1 \cdots B_{m-1}} \leq H_{\lambda}.
\end{align}

In order to prove the inequality $BI^{E}_{mm} \leq \mathcal{C}$ \eqref{BI_H}, we first note that the rows $j=2,\cdots,m-2$ are exactly the same as for the $I^{E}_{mm} \leq 0$ inequality \eqref{Imm_ent}, that fulfill the
inequality
\begin{equation}
-\sum_{j=2,\cdots,m-2} \sum_{i=0,\cdots,m-j} (-1)^{\delta_{i,m-j}} H_{A_i B_j} + (m-j-1)H_{B_j} + H_{A_0 A_1} - H_{A_0 A_1 \cdots A_{m-2}} \leq 0.
\end{equation}
If we can prove that the first two rows plus the last row plus the marginal $H_{A_0}$ in \eqref{BI_H} define a valid inequality of the form
\begin{align}
& H_{B_0}-H_{A_{m-1} B0} -\sum_{i=0,\cdots,m-1} (-1)^{\delta_{i,m-1}} H_{A_i B_j} + (m-2)H_{B_1} \\ \nonumber
& +H_{A_0} -H_{A_0 B_{m-1}}+H_{A_1 B_{m-1}} - H_{A_0 A_1} + H_{A_0 A_1 \cdots A_{m-2}} \leq  H_{\lambda} \leq \mathcal{C}
\end{align}
then we have proven $BI^{E}_{mm}\leq \mathcal{C}$.

We are going to show next the proof for $m$ even but $m$ odd follows along similar lines.

First add
\begin{equation}
-H_{A_0 B_{m-1}}-H_{A_0 A_1}+H_{A_1 B_{m-1}}+H_{A_0} \leq 0.
\end{equation}
We continue adding the following SM inequalities
\begin{align}
& -H_{A_0 B_1}-H_{A_1 B_1}+H_{A_0 A_1 B_1}+H_{B_1} \leq 0 \\
& -H_{A_2 B_1}-H_{A_3 B_1}+H_{A_2 A_3 B_1}+H_{B_1} \leq 0\\\nonumber
&\vdots \\
& -H_{A_{m-4} B_1}-H_{A_{m-3} B_1}+H_{A_{m-4} A_{m-3} B_1}+H_{B_1} \leq 0.
\end{align}

Add
\begin{align}
& H_{A_0 A_1 \cdots A_{m-1}}+H_{B_0 B_1}-H_{A_0 \cdots A_{m-1} B_0 B_1} - H_{\lambda} \leq 0 \\
& H_{A_{m-1} B_0 B_1}+H_{B_0}-H_{B_0 B_1}-H_{A_{m-1} B_0} \leq \\
& H_{A_0 A_1 \cdots A_{m-1} B_0 B_1}+H_{A_{m-1} B_1}-H_{A_0 A_1 \cdots A_{m-1} B_1}-H_{A_{m-1} B_0 B_1} \leq 0 \\
& H_{A_0 A_1 \cdots A_{m-2}}+H_{A_0 A_1 \cdots A_{m-1} B_1}-H_{A_0 A_1 \cdots A_{m-1}}-H_{A_0 A_1 \cdots,A_{m-2} B_1} \leq 0.
\end{align}

At this point the missing terms are (plus the NO terms)
\begin{equation}
\frac{(m-2)}{2} H_{B_1}-H_{A_{m-2} B_1}-H_{A_0 A_1 \cdots A_{m-2} B_1}+\sum_{i=0,\cdots,m-4}H_{A_{i} A_{i+1} B_1}  \leq 0.
\end{equation}

Add
\begin{align}
& H_{A_0 A_1 \cdots A_{m-2} B_1}+H_{B_1}-H_{A_2 A_3 \cdots A_{m-2} B_1}-H_{A_0 A_1 B_1}  \leq 0\\
& H_{A_2 A_3 \cdots A_{m-2} B_1}+H_{B_1}-H_{A_4 A_5 \cdots A_{m-2} B_1}-H_{A_2 A_3 B_1} \leq 0\\ \nonumber
& \vdots \\
& H_{A_{m-4} \cdots A_{m-2} B_1}+H_{B_1}-H_{A_{m-2} B_1}-H_{A_{m-4} A_{m-5} B_1} \leq 0.
\end{align}

This concludes the proof.

\end{widetext}

\end{document}